%% file: main.tex
\documentclass[sigplan,10pt]{acmart}
\AtBeginDocument{%
  }

\acmYear{2026}\copyrightyear{2026}
\setcopyright{cc}
\setcctype[4.0]{by}
\acmConference[EUROSYS '26]{European Conference on Computer Systems}{April 27--30, 2026}{Edinburgh, Scotland Uk}
\acmBooktitle{European Conference on Computer Systems (EUROSYS '26), April 27--30, 2026, Edinburgh, Scotland Uk}
\acmDOI{10.1145/3767295.3769342}
\acmISBN{979-8-4007-2212-7/26/04}
\acmSubmissionID{337}

\input{preamble}
\begin{document}

\title{OptiLog: Assigning Roles in Byzantine Consensus}

\author{Hanish Gogada}
\affiliation{%
    \institution{University of Stavanger}
    \country{Norway}
}
\email{hanish.gogada@uis.no}

\author{Christian Berger}
\affiliation{%
    \institution{Friedrich-Alexander-Universität Erlangen-Nürnberg}
    \country{Germany}
}
\email{christian.g.berger@fau.de}

\author{Leander Jehl}
\affiliation{%
    \institution{University of Stavanger}
    \country{Norway}
}
\email{leander.jehl@uis.no}

\author{Hans P. Reiser}
\affiliation{%
    \institution{Reykjavik University}
    \country{Iceland}
}

\author{Hein Meling}
\affiliation{%
    \institution{University of Stavanger}
    \country{Norway}
}

\input{tex/abstract}



\maketitle
\input{tex/introduction}

\input{tex/preliminaries}
\input{tex/motivation}

\input{tex/optilog}

\input{tex/monitors_intro}

\input{tex/monitor_latency}

\input{tex/monitor_misbehavior}

\input{tex/monitor_suspicion}
\input{tex/monitor_config}

\input{tex/aware.tex}

\input{tex/optitree}
\input{tex/optitree_monitors}

\input{tex/evaluation_intro}
\input{tex/evaluation_baseline}

\input{tex/security_latency}

\input{tex/evaluation_tolerance}

\input{tex/evaluation_sa}

\input{tex/evaluation_overhead}
\input{tex/related}

\input{tex/conclusion}

\input{tex/acknowledgements}

\bibliographystyle{ACM-Reference-Format}
\bibliography{references}

\appendix
\input{tex/notation}
\input{tex/evaluation_score_app.tex}

\input{tex/evaluation_reconfig_app.tex}
\balance
\input{tex/suspicion_proof_app.tex}
\input{tex/security_analysis_app.tex}
\input{tex/reconfig_proof_app.tex}

\end{document}

%% file: preamble.tex
\usepackage{amsthm}
\usepackage{amsmath}
\usepackage{amsfonts}
\usepackage{graphicx}
\usepackage[capitalize,nameinlink]{cleveref}
\usepackage{acronym}
\usepackage{xspace}
\usepackage{subcaption}
\usepackage{pgfplots}
\usepackage{pgfplotstable}
\pgfplotsset{compat=1.18}
\usepgfplotslibrary{groupplots}
\usetikzlibrary{arrows,trees,patterns,positioning,shapes,fit}
\usepackage{booktabs}
\usepackage{algorithm}
\usepackage[noend]{algpseudocode}
\usepackage{listings}
\usepackage{url}
\usepackage{pifont}
\usepackage{tabularx}
\usepackage{enumitem}
\usepackage[group-separator={,}]{siunitx}
\microtypecontext{spacing=nonfrench}

\hypersetup{breaklinks=true,
  pdftitle={OptiLog},
  colorlinks=true,
  citecolor=black,
  urlcolor=black,
  linkcolor=black,
  pdfborder={0 0 0},
}

\newboolean{showcomments}
\setboolean{showcomments}{true}
\ifthenelse{\boolean{showcomments}}
{ \newcommand{\mynote}[3]{
    \textcolor{#3}{{\bfseries\sffamily\scriptsize#1: }\small#2}
}}
{ \newcommand{\mynote}[3]{}}

\newcommand{\itz}[1]{\textit{#1}}

\definecolor{revisioncolor}{rgb}{0.3,0.7,0.5}

\crefformat{section}{#2§#1#3}
\crefname{lstlisting}{Listing}{Listings}

\makeatletter
\AtBeginDocument{
  \def\ltx@label#1{\cref@label{#1}}%
  \def\label@in@display@noarg#1{\cref@old@label@in@display{#1}}%
  \def\label@in@mmeasure@noarg#1{%
    \begingroup%
      \measuring@false%
      \cref@old@label@in@display{#1}%
    \endgroup}%
}
\makeatother

\acrodef{BFT}{Byzantine Fault-Tolerant}
\acrodef{RSM}{Replicated State Machine}
\acrodef{GST}{\textit{global stabilization time}}

\newcommand{\sysname}{Opti\-Log\xspace}
\newcommand{\Sysname}{Opti\-Log\xspace}

\newcommand{\consmod}{consensus engine\xspace}

\newcommand{\sensapp}{sensor app\xspace}

\newcommand{\Monapps}{Monitors\xspace}

\newcommand{\Sensors}{Sensors\xspace}
\newcommand{\sensor}{sensor\xspace}
\newcommand{\sensors}{sensors\xspace}
\newcommand{\missensor}{Misbehavior\-Sensor\xspace}
\newcommand{\latsensor}{Latency\-Sensor\xspace}
\newcommand{\sussensor}{Suspicion\-Sensor\xspace}
\newcommand{\cfgsensor}{Config\-Sensor\xspace}

\newcommand{\monitor}{monitor\xspace}
\newcommand{\monitors}{monitors\xspace}
\newcommand{\mismonitor}{Mis\-be\-havior\-Monitor\xspace}
\newcommand{\latmonitor}{Latency\-Monitor\xspace}
\newcommand{\susmonitor}{Suspicion\-Monitor\xspace}
\newcommand{\cfgmonitor}{Config\-Monitor\xspace}
\newcommand{\validt}{correct\xspace}

\newcommand{\configsensor}{configuration \sensor\xspace}

\newcommand{\optitree}{Opti\-Tree\xspace}
\newcommand{\oware}{Opti\-Aware\xspace}
\newcommand{\tbc}{$t$-Bounded Conformity\xspace}

\newcommand{\node}{\ensuremath{\mathit{node}}}

\newcommand{\signed}[2]{\ensuremath{\langle m \rangle_i}}
\newcommand{\score}{\texttt{score}\xspace}
\newcommand{\scorefn}{\texttt{score($\cdot$)}\xspace}
\newcommand{\susp}[2]{\ensuremath{#1\!\dashrightarrow\!#2}\xspace}
\newcommand{\susptwo}[2]{\ensuremath{#1\!\dashleftrightarrow\!#2}\xspace}
\newcommand{\suspm}[3]{\ensuremath{\langle \textsc{#1},\ \susp{#2}{#3} \rangle}\xspace}
\newcommand{\latm}[3]{\ensuremath{\langle \textsc{#1},\ \lat{#2}{#3} \rangle}\xspace}
\newcommand{\compm}{\ensuremath{\langle \textsc{Complaint},\ B \rangle}\xspace}
\newcommand{\Cfg}{\ensuremath{\mathit{\Gamma}}}
\newcommand{\cfgm}{\ensuremath{\langle \textsc{Config},\ \Cfg_A \rangle}\xspace}

\theoremstyle{definition}
\newtheorem{definition}{Definition}
\newtheorem{thm}{Theorem}
\newtheorem{lemma}[thm]{Lemma}
\AddToHook{env/lemma/begin}{\crefalias{thm}{lemma}}

\algrenewcommand{\algorithmiccomment}[1]{\hfill$\{$#1$\}$}

\algblockdefx[On]{On}{EndOn}
[1]{{\bf On} #1}
{EndOn}
\algnotext{EndOn}

\algsetblockdefx[State]{STate}{ESTate}
{}{0pt}
[1]{{\bf State} #1}
{EndState}
\algnotext{ESTate}

\algblockdefx[Init]{Init}{EInit}
[1]{{\bf init}$(#1)$}
{EndInit}
\algnotext{EInit}

\newcommand{\n}{\ensuremath{\mathit{n}}\xspace}      %
\newcommand{\q}{\ensuremath{\mathit{q}}\xspace}      %
\newcommand{\f}{\ensuremath{\mathit{f}}\xspace}      %
\newcommand{\tn}{\ensuremath{\mathit{t}}\xspace}     %
\newcommand{\bn}{\ensuremath{\mathit{b}}\xspace}     %
\newcommand{\kn}{\ensuremath{\mathit{k}}\xspace}     %
\newcommand{\un}{\ensuremath{\mathit{u}}\xspace}     %
\newcommand{\inn}{\ensuremath{\mathit{i}}\xspace}    %
\newcommand{\winlen}{\ensuremath{\mathit{w}}\xspace} %

\newcommand{\na}{\ensuremath{\mathit{A}}\xspace}     %
\newcommand{\nb}{\ensuremath{\mathit{B}}\xspace}     %
\newcommand{\nc}{\ensuremath{\mathit{C}}\xspace}     %
\newcommand{\nd}{\ensuremath{\mathit{D}}\xspace}     %
\newcommand{\nv}{\ensuremath{\mathit{V}}\xspace}     %
\newcommand{\nr}{\ensuremath{\mathit{R}}\xspace}     %
\newcommand{\nin}{\ensuremath{\mathit{I}}\xspace}    %
\newcommand{\nleaf}[1]{\ensuremath{\mathit{T}_{#1}}\xspace}  %
\newcommand{\nl}{\ensuremath{\mathit{L}}\xspace}     %
\newcommand{\naT}{\ensuremath{\mathit{A_t}}\xspace}  %
\newcommand{\nbC}{\ensuremath{\mathit{B_c}}\xspace}  %
\newcommand{\cn}[1]{\ensuremath{\mathit{N}_{#1}}\xspace}  %
\newcommand{\susn}[1]{\ensuremath{\mathit{S}_{#1}}\xspace}  %

\newcommand{\sens}[1]{\ensuremath{\mathsf{S}_{#1}}\xspace}  %
\newcommand{\mons}[1]{\ensuremath{\mathsf{M}_{#1}}\xspace}  %

\newcommand{\df}[1]{\ensuremath{\delta \cdot #1}}               %
\newcommand{\rdur}{\ensuremath{d_{rnd}}}                        %
\newcommand{\dm}{\ensuremath{d_m}}                              %

\newcommand{\lm}{\ensuremath{\mathbf{L}}\xspace}             %
\newcommand{\lme}[2]{\ensuremath{\mathbf{L}_{#1,#2}}\xspace} %

\newcommand{\C}[1]{\ensuremath{\Pi_{#1}}\xspace}             %
\newcommand{\Q}{\ensuremath{\mathcal{Q}}\xspace}             %
\newcommand{\T}{\ensuremath{\mathcal{T}}\xspace}             %
\newcommand{\Crash}{\ensuremath{\mathcal{C}}\xspace}         %
\newcommand{\Internal}{\ensuremath{\mathcal{I}}\xspace}      %
\newcommand{\Leaf}{\ensuremath{\mathcal{L}}\xspace}          %
\newcommand{\Intermediate}{\ensuremath{\mathcal{M}}\xspace}  %
\newcommand{\tree}{\ensuremath{\tau}\xspace}                 %
\newcommand{\IS}{\ensuremath{IS}\xspace}                     %
\newcommand{\Ch}[1]{\ensuremath{\textnormal{Ch}(#1)}\xspace} %
\newcommand{\G}{\ensuremath{\mathcal{G}}\xspace}             %
\newcommand{\V}{\ensuremath{\mathcal{V}}\xspace}             %
\newcommand{\E}{\ensuremath{\mathcal{E}}\xspace}             %
\newcommand{\MG}{\ensuremath{\E_d}\xspace}                   %
\newcommand{\Cand}{\ensuremath{\mathcal{K}}\xspace}          %
\newcommand{\Faulty}{\ensuremath{\mathcal{F}}\xspace}        %

\newcommand{\lat}[3][r]{\ensuremath{\mathit{L_{#1}(#2,#3)}}}     %
\newcommand{\agt}[1]{\ensuremath{\mathit{L_{agg}(#1)}}}          %
\newcommand{\sqn}{\ensuremath{\sqrt{\n}}}                        %

\definecolor{dark-blue}{rgb}{0.18, 0.33, 0.59}
\definecolor{light-blue}{rgb}{0.56, 0.67, 0.86}

\definecolor{dark-red}{rgb}{0.49, 0, 0}
\definecolor{light-red}{rgb}{0.75, 0, 0}

\definecolor{dark-orange}{rgb}{0.77, 0.35, 0.07}
\definecolor{light-orange}{rgb}{0.92, 0.49, 0.19}

\definecolor{dark-yellow}{rgb}{1, 0.75, 0}
\definecolor{light-yellow}{rgb}{1, 0.85, 0.4}

\definecolor{dark-green}{rgb}{0.66, 0.81, 0.56}
\definecolor{light-green}{rgb}{0.77, 0.88, 0.71}

\definecolor{db1}{rgb}{0.08,0.14,0.25}
\definecolor{lb1}{rgb}{0.12,0.22,0.39}

\definecolor{db2}{rgb}{0.23,0.3,0.44}
\definecolor{lb2}{rgb}{0.21,0.32,0.56}

\definecolor{db3}{rgb}{0.29,0.48,0.64}
\definecolor{lb3}{rgb}{0.51,0.58,0.73}

\definecolor{db4}{rgb}{0.55,0.67,0.86}
\definecolor{lb4}{rgb}{0.71,0.78,0.91}

%% file: tex/abstract.tex
\begin{abstract}
Byzantine Fault-Tolerant (BFT) protocols play an important role in blockchains.
As the deployment of such systems extends to wide-area networks, the scalability of BFT protocols becomes a critical concern.
Optimizations that assign specific roles to individual replicas can significantly improve the performance of BFT systems.
However, such role assignment is highly sensitive to faults, potentially undermining the optimizations' effectiveness.

To address these challenges, we present \sysname, a logging framework for collecting and analyzing measurements that help to assign roles in globally distributed systems, despite the presence of faults.
\sysname presents local measurements in global data structures, to enable consistent decisions and hold replicas accountable if they do not perform according to their reported measurements.

We demonstrate \sysname's flexibility by applying it to two BFT protocols:
(1)~Aware, a highly optimized PBFT-like protocol, and
(2)~Kauri, a tree-based protocol designed for large-scale deployments.
\sysname detects and excludes replicas that misbehave during consensus and thus enables the system to operate in an optimized, low-latency configuration, even under adverse conditions.
Experiments show that for tree overlays deployed across 73 worldwide cities, trees found by \sysname display 39\% lower latency than Kauri.
\end{abstract}

%% file: tex/introduction.tex
\section{Introduction}

\ac{BFT} consensus protocols enable distributed systems that can tolerate arbitrary failures, making them a critical component for blockchain systems that distribute trust among participants that validate transactions and maintain the ledger.
Therefore, to reduce reliance on a few participants and bolster security, scaling BFT protocols beyond their traditional capacity~\cite{pbft, zyzzyva} has become crucial.

Numerous studies have tackled the scalability challenges using various optimization strategies effective under specific favorable conditions~\cite{sbft, algorand, tendermint, kauri, mirbft}.
However, in adverse scenarios, these optimizations become less effective or even defunct, necessitating a switch~\cite{Abstract, adapt} to a more resilient but less effective protocol, like PBFT~\cite{pbft}, to make progress.
It is, however, non-trivial to detect whether or not the operating conditions are favorable or adversarial.
Existing work utilizes inefficient techniques based on randomness~\cite{algorand, kauri, mytumbler} or predefined assignments~\cite{mirbft}, requiring repeated trial-and-error to find a working system configuration.
These systems often ignore the actual operating conditions, such as the latency between replicas and prior misbehavior, resulting in poor performance.
Many systems~\cite{bft-smart,local1,local2,local3,eval} that consider the operating conditions rely only on local measurements to select a system configuration.
This is problematic because local measurements taken at different replicas may be inconsistent, making it difficult to derive a global configuration.
Moreover, delegating configuration decisions to a single replica, e.g., the leader, based on its own measurements, erodes accountability because other replicas cannot verify that the decision was based on actual measurements.

This paper advocates for a holistic, measurement-based approach to accurately identify the current operating conditions, promoting aggressive use of efficient protocols and reducing fallback to less efficient ones.
We accomplish this through a shared append-only log of measurements.

\begin{figure}[t]
  \centering
  \includegraphics[width=\columnwidth]{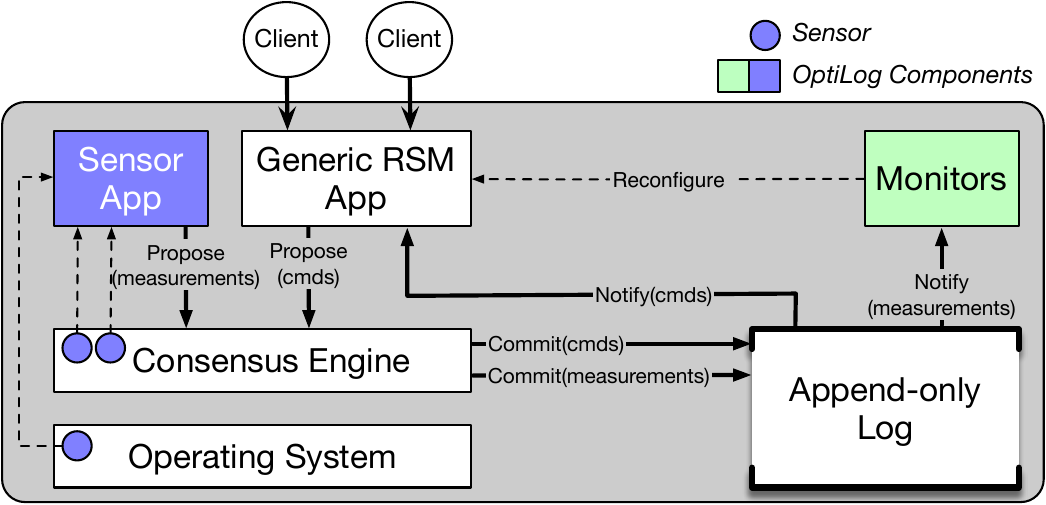}
  \caption{\sysname's component architecture.}
  \Description{Diagram showing the architecture of \sysname, including the sensor application, sensors and monitors.}
  \label{fig:arch}
\end{figure}

We describe \sysname, an integrated shared log that records measurements and derives metrics for system configuration.
\sysname extends a generic \ac{RSM} with sensors to capture \textit{measurements} and monitors for evaluating \textit{metrics}.
Replicas instrumented with sensors append measurements to the log, while corresponding monitors collate these measurements to derive \textit{efficient configurations}.
\sysname enables replicas to make consistent configuration decisions based on the same information.
The log also provides accountability, allowing all replicas to verify decisions and recognize faulty behaviors.
\cref{fig:arch} illustrates \sysname's architecture; see \cref{sec:optiblogarch} for a detailed description.

Some optimization techniques can be costly to (deterministically) evaluate for large configuration sizes.
Thus, \sysname allows heuristic optimization techniques to be used, where the potentially non-deterministic search outcomes are logged, enabling replicas to consistently derive a global ranking of configurations.
\sysname also supports collaborative optimization techniques, where the search space is partitioned and distributed across the \ac{RSM} replicas.

Our second contribution is a pipeline of sensors and monitors for selecting functional, low-latency configurations.
Our pipeline allows replicas to predict a configuration's latency based on individual link latencies, and to use simulated annealing~\cite{simulatedannealing} to search for low-latency configurations.
The main novelty of our pipeline is a suspicion mechanism that detects timing failures and takes suspicions into account during configuration search.
In particular, our suspicion mechanism detects replicas that underperform relative to their self-reported latencies.
This allows \sysname to find efficient configurations even in the presence of performance attacks~\cite{prime,clement2009making}.

To showcase \sysname, we integrate our pipeline into two BFT protocols: Aware~\cite{aware} and Kauri~\cite{kauri}.
Aware and Kauri are based on PBFT~\cite{pbft} and HotStuff~\cite{hotstuff}, respectively, and represent optimizations of different classes of BFT protocols, each posing distinct challenges.
This demonstrates the flexibility and generality of \sysname.

\textbf{\oware.}
Aware is a modern variant of PBFT and assumes a clique topology with multiple all-to-all communication phases, where replicas respond after collecting messages from a quorum.
Each replica assigns weights to other replicas based on the measured latency between them.
The all-to-all message pattern allows any configuration with a functioning leader to remain operational, despite backup failures.
In PBFT however, distinguishing benign from intentional delays is hard, making it tricky to set a detection threshold that avoids false positives but still catches real delays.
While Aware improves on PBFT's latency, it remains vulnerable to performance attacks.
Our third contribution, \oware, applies \sysname's pipeline to Aware, yielding a system that can detect the culprits of such attacks and mitigate their impact by excluding them from consideration when assigning high voting weights and the leader role.

\textbf{\optitree.}
For our fourth contribution, we focus on BFT systems~\cite{kauri, omniledger, byzcoin} that organize their RSM replicas in a tree topology for improved scalability.
A tree topology reduces communication overhead by limiting interactions to parent and child replicas.
However, this structure is vulnerable to failures, which can fragment the network.
To counter this, tree-based protocols require fallback and reconfiguration mechanisms.
Kauri~\cite{kauri}, for example, builds multiple randomized trees to prevent targeted attacks, switching to a new tree when failures occur.
However, this approach can cause considerable delays in reaching consensus due to the number of reconfigurations needed to find a working tree.
Moreover, blindly constructing trees based on randomness can result in latency-imbalanced trees that degrade the RSM's performance.
Ideally, trees should be constructed with reliable replicas at the core and less reliable ones at the leaves.

\optitree applies \sysname's sensor and monitor pipeline to find working and performant trees.
We implemented two variants of \optitree and compared their throughput under various latency conditions.
Our experiments show that trees selected by \optitree outperform Kauri's throughput by up to $2.5\times$ and HotStuff by $2.9\times$.
Further, \optitree can reconfigure the RSM until a working tree is found, while recording the measurements in the log.
Our approach guarantees to find a working tree within at most 2\f reconfigurations for certain tree configurations, where \f is the maximum number of faulty replicas.

In summary, the main contributions of this paper are:
\begin{enumerate}%
  \item \sysname, a logging framework for systematic measurement collection and configuration optimization for RSMs.
  \item \sysname enables RSM optimizations to compute functional low-latency configurations in the presence of faults.
  \item \oware uses \sysname to detect-and-mitigate performance attacks and maintain optimal latency.
  \item \optitree uses \sysname to generate low-latency working trees with linear reconfigurations.
  \item Trees selected by \optitree boost Kauri's throughput by 67.5\% in a simulated Stellar network~\cite{stellar}.
\end{enumerate}

%% file: tex/preliminaries.tex
\section{System Model}
\label{sec:prelim}
We consider an \ac{RSM} that accepts client commands as input and produces output to clients.
The RSM may record certain events in an append-only log.
We assume an RSM with $\n\ge3\f+1$ replicas, where up to $\f$ can be Byzantine faulty, meaning they may exhibit arbitrary, potentially adversarial behavior.
Replicas that faithfully follow the protocol and do not crash are considered \textit{correct}, while others are deemed \textit{faulty}.
In every view, at least a quorum of $\q=\n-\f$ correct replicas is assumed to be available.

We assume a \textit{partially synchronous} network~\cite{partialsynchronous}, where periods of instability may cause unpredictable delays.
However, after a \ac{GST}, latencies become bounded.
We also assume there exists a known parameter $\delta$ and an \textit{actual} latency $\lat[a]{\na}{\nb}$ between correct replicas $\na$ and $\nb$, such that after GST, the round-trip time of any message between $\na$ and $\nb$ is in the interval $[\lat[a]{\na}{\nb}, \df{\lat[a]{\na}{\nb}}]$.

Note that the partial synchrony assumption only pertains to liveness; the RSM's safety is always preserved.
A faulty replica may attempt to disrupt the RSM's performance by recording incorrect measurements in the log.
However, an adversary cannot delay messages, including measurements, between correct replicas.
Thus, faulty replicas cannot interfere with measurement exchanges between correct replicas.

Let $\C{}$ denote the set of $\n$ replicas running the RSM. A \textit{configuration} is an assignment of roles to replicas, which may also encode topology information, e.g., a tree configuration.
While describing a replica's position in a tree/graph, it may also be referred to as a \textit{node}.
An RSM may \textit{reconfigure} from one configuration to another, e.g., to activate a better configuration or recover from a faulty configuration.

%% file: tex/motivation.tex
\section{Motivation}
\label{sec:motivation}

\begin{quote}
\textit{If you can't measure it, you can't improve it.}

\hfill ---Lord Kelvin
\end{quote}

\noindent
This section outlines several challenges with existing systems that motivate \sysname's design.

\noindent\textbf{Performance tuning.}
Optimizing the performance of RSM-based protocols necessitates precise tuning to the conditions of their deployment environment.
Performance parameters---such as the replication degree, batch size, concurrency, shard count, and network timeouts---must be carefully calibrated to balance efficiency and reliability.
For instance, RSMs operating over wide-area networks with inherently higher and less predictable latencies require more conservative timeouts compared to those on local area networks~\cite{quepaxa,atlas}.
For leader-based protocols, choosing a short timeout may lead to consecutive leader changes~\cite{mytumbler}, while a longer timeout results in slower recovery from failures.
A single leader often becomes a performance bottleneck~\cite{hovercraft,scalablesmr}, and rotating the leader-role~\cite{hotstuff} can result in significant performance degradation when some replicas are faulty~\cite{mytumbler,beegees}.
To the best of our knowledge, real-world RSM deployments lack real-time feedback mechanisms to optimize their configuration.

\noindent\textbf{Limitations of local-only measurements.}
Local measurements~\cite{local1,local2,eval}, where individual replicas capture and use measurements without sharing them with other replicas, suffer from inherent limitations.
Such local-only observations can lead to discrepancies, as replicas may draw inconsistent conclusions about their operating environment.
Furthermore, replicas cannot be held accountable for configuration decisions based on local measurements, since such decisions cannot be verified by other replicas.

\noindent\textbf{Measurement-based role assignment.}
A wide range of strategies exists for RSM optimization~\cite{berger2023sok}, including pipelining~\cite{rcc,sbft}, committee selection~\cite{dumbo,proteus,proof-of-qos}, and tree-based communication topologies~\cite{kauri, byzcoin}.
These optimizations rely on special roles, like coordinators~\cite{rcc}, collectors~\cite{sbft}, committee members~\cite{proof-of-qos} and internal nodes~\cite{kauri}.
These roles are assigned in a randomized~\cite{algorand, kauri} or predefined fashion~\cite{zorfu, mirbft, mencius}, requiring trial-and-error to find a working assignment.
This precludes informed decisions and learning from past configuration failures.

Performing end-to-end measurements to evaluate and select configurations~\cite{archer} becomes impractical as the configuration search space grows.
For example, when assigning voting weights~\cite{wheat} or positions in a tree overlay~\cite{kauri}, the number of possible configurations is exponential in the number of replicas, making it infeasible to perform end-to-end measurements for all configurations.

Aware~\cite{aware} selects low-latency configurations based on per-link latency measurements and a scoring function.
While this approach is applicable to large configuration search spaces, Aware lacks a mechanism to hold misbehaving replicas accountable when configurations do not perform as expected.
In particular, if probe messages exhibit different latencies than actual protocol messages, Aware is unable to select a suitable configuration, as we show in \cref{eval:oware:underAttack}.

Some systems~\cite{archer, clement2009making} use measurements to monitor configuration performance, while others rely on observations of past failures~\cite{spinning, prosecutor}.
However, these systems focus on leader selection and hold only the leader accountable for failed or degraded configurations.
This limitation makes them unsuitable for complex configurations, such as tree topologies~\cite{kauri}, which may fail despite a correct leader.

\noindent\textbf{Collaborative optimization.}
Optimizing RSM configurations by exploring the search space on a single replica creates a performance bottleneck.
Throughput can be improved by partitioning the search space and distributing the partitions across replicas via scatter-gather~\cite{mpi} or map-reduce~\cite{map-reduce, spark} techniques.
Yet, these distribution techniques remain surprisingly underutilized for optimizing RSM configurations.

\noindent\textbf{Non-deterministic configuration optimization.}
Finding an effective configuration can be framed as a combinatorial optimization problem.
Due to the vastness of possible configurations, finding an optimal solution is often unfeasible.
This makes heuristic approaches, including machine learning, more viable than deterministic methods.
However, due to the inherent non-determinism of these heuristic solutions, existing systems~\cite{aware} do not support non-deterministic configuration optimization.

%% file: tex/optilog.tex
\section{\sysname Architecture Overview}
\label{sec:optiblogarch}

This section gives an overview of \sysname's architecture, as shown in \cref{fig:arch}.
\sysname provides measurement and monitoring capabilities to optimize the performance and scalability of generic RSM applications.
The RSM uses a \textit{\consmod} to consistently replicate client commands to an \textit{append-only log}.
Optimizing the RSM requires metrics on both its state and environment.
To achieve this, \sysname augments each replica with \textit{\sensors} and \textit{\monitors}.
These components use the log to share information about the system's configuration and sensor measurements with all replicas.

\Sensors collect various measurements and pass them to the \textit{\sensapp}, which uses the \consmod to disseminate authenticated measurements to all replicas.
These measurements can come from different sources, such as the replica's operating environment or sensors integrated into the \consmod.
The \monitors act as counterparts to \sensors, receiving notifications about specific measurements and use them to compute various metrics.

\cref{fig:sensors-and-monitors} shows how monitor \mons{1} (\mons{2}) collect measurements from sensor \sens{1} (\sens{2}) at replicas \na, \nb, \nc, and \nd.
This allows \sysname replicas to make operational decisions based on deterministic information rather than local-only sensor data.
As shown in \cref{fig:sensors-and-monitors}, monitor \mons{2} makes reconfiguration decisions for the RSM based on globally recorded and agreed-upon measurements.
For example, one of our \monitors generates a latency matrix that provides insights into the latencies between all participating replicas, as discussed in \cref{sec:latency}.

%% file: tex/monitors_intro.tex
\subsection{\Sensors and \Monapps}
\label{sec:sm}

This section presents \sysname's mechanisms for collecting and processing measurements, namely \sensors and \monitors.
Each sensor-monitor pair has a one-to-one relationship, as shown with \sens{1}--\mons{1} and \sens{2}--\mons{2} in \cref{fig:sensors-and-monitors}.

\begin{figure}[t]
  \centering
  \input{images/sensors-and-monitors}
  \caption{Replicas (\na,\nb,\nc,\nd) with sensors and monitors.}
  \Description{Diagram illustrating the relationship between sensors and monitoring applications in \sysname.}
  \label{fig:sensors-and-monitors}
\end{figure}

A \textit{\sensor} abstracts the capture of local-only measurements related to a replica's operational context.
\Sensors can be integrated into system components, such as the \consmod, or can query the operating system for metrics such as CPU load.
A sensor can also use information from local monitors, as exemplified by the dashed arrow from \mons{1} to \sens{1} in \cref{fig:sensors-and-monitors}.
A \sensor's output is assumed to be non-deterministic and may vary across replicas.
\Sensors may also perform non-deterministic computations, e.g., to support heuristic and distributed optimization algorithms.
We refer to both sensor output and compute results as \textit{measurements}.
These measurements are recorded in the log, providing a consistent record of the system's operation.
The recorded measurements are later accessed by \monitors, as discussed below.

In a Byzantine environment, \sensors at up to \f faulty replicas may report incorrect measurements.
Consequently, any \monitor using these measurements must account for potential inaccuracies.
Yet retaining a record of incorrect measurements can be invaluable for forensic analysis.

The \textit{\monitor} is the core abstraction linking sensor measurements from individual replicas with system-wide operations, enabling dynamic adjustments to the RSM.
When a sensor records a measurement, the corresponding \monitor updates its data structures through the following steps:

\begin{enumerate}
  \item \textbf{Data Collection:}
    The monitor collects measurements from its associated sensor via the log; it may also use data from other local monitors.
  \item \textbf{Deterministic Computation:}
    The monitor processes consistent measurements by performing a deterministic computation,
    producing consistent output; referred to as metrics.
\end{enumerate}

\noindent
Since monitors on each replica operate on the same ordered set of measurements, they maintain consistent data structures across the system.
This consistency provides a consolidated, global view of the measurements, allowing replicas to coordinate system-wide operations reliably, such as activating a new tree configuration with \optitree~(\cref{sec:optitree}).
\cref{tab:sensors_monitors} summarizes the key properties of sensors and monitors.

\begin{table}[htb]
  \footnotesize
  \centering
  \caption{Summary of sensor and monitor properties.}
  \label{tab:sensors_monitors}
  \begin{tabular}{@{}p{1.8cm}|p{2.9cm}p{2.9cm}@{}}
    \toprule
    \textbf{Property}        & \textbf{Sensors}          & \textbf{Monitors}          \\ \midrule
    \textbf{Input}           & System \& local monitors  & Log \& local monitors      \\
    \textbf{Computation}     & Non-deterministic         & Deterministic              \\
    \textbf{Output}          & Variable across replicas  & Consistent across replicas \\
    \bottomrule
  \end{tabular}
\end{table}

\begin{figure*}[bt]
  \centering
  \input{images/star-sens-mons}
  \caption{Connections between Replica \na's sensors and monitors.}
  \Description{Diagram illustrating the connections between sensors and monitors for a star topology.}
  \label{fig:star-sens-mons}
\end{figure*}

\noindent
In the following sections we present the sensors and monitors we implemented for \sysname, and \cref{fig:star-sens-mons} shows how these interact to ensure robust configurations.

%% file: images/sensors-and-monitors.tex
\def\smsize{0.6cm}
\def\smpos{1}
\def\rsize{1cm}
\def\dist{1.5}
\def\logpos{0}
\def\logh{0.5}
\def\logw{8.9}
\def\arrowspacing{4pt}
\def\legendspacing{0.28}

\definecolor{darkcola}{RGB}{140, 140, 180} %
\definecolor{cola}{RGB}{230, 230, 250} %
\definecolor{colb}{RGB}{100, 200, 100} %
\definecolor{colc}{RGB}{200, 100, 100} %
\definecolor{cold}{RGB}{100, 180, 180} %

\tikzset{
  arrow-record-a/.style={->, thick, color=#1},
  arrow-record-b/.style={->, thick, >={latex'}, color=#1},
  arrow-notify-a/.style={<-, thick, color=#1},
  arrow-notify-b/.style={<-, thick, >={latex'}, color=#1},
  arrow-local/.style={->, thick, dashed},
  arrow-sh/.style={xshift=#1*\arrowspacing},
}

\tikzset{
  label/.style={midway, yshift=-2pt, font=\scriptsize, fill=white, fill opacity=1, text opacity=1, inner sep=1.5pt},
  legend/.style={left, font=\scriptsize},
  box/.style={draw, rectangle, minimum width=\smsize, minimum height=\smsize, pattern=crosshatch, pattern color=white},
  replica/.style={draw, rectangle, minimum width=\rsize, minimum height=\smsize},
  sensor-box/.style={box, preaction={fill=#1}},
  monitor-box/.style={box, preaction={fill=#1}},
  replica-box/.style={replica, preaction={fill=#1}},
}

\begin{tikzpicture}

  \draw[thick, fill=cola] (0.5, 1.9) node[anchor=north west] {Replica \na} rectangle (\logw, \logpos+0.5);

  \node[sensor-box=light-yellow] (S1) at (1*\dist, \smpos) {$\sens{1}$};
  \node[monitor-box=light-yellow] (M1) at (2*\dist, \smpos) {$\mons{1}$};
  \node[sensor-box=light-blue] (S2) at (3*\dist, \smpos)   {$\sens{2}$};
  \node[monitor-box=light-blue] (M2) at (4*\dist, \smpos)   {$\mons{2}$};
  \node[replica-box=light-blue] (RSM) at (5.5*\dist, \smpos) {RSM};

  \draw[arrow-local] (M2.east) -- (RSM.west) node[midway, yshift=-3pt, font=\scriptsize, above] {Reconfigure};

  \draw[arrow-local] (M1) -- (S1);
  \draw[arrow-local] (M1) -- (S2);
  \draw[arrow-local] (M1.north) to[out=60, in=120, looseness=0.5] (M2.north);

  \draw[arrow-record-a=darkcola] (S1.south) -- (1*\dist, \logpos) node[label] {Record};
  \draw[arrow-record-b=darkcola] (S2.south) -- (3*\dist, \logpos) node[label, yshift=1.5pt] {Record};

  \draw[arrow-notify-a=colb] ([arrow-sh=-1.5]M1.south) -- ([arrow-sh=-1.5]2*\dist, \logpos);
  \draw[arrow-notify-a=colc] ([arrow-sh=-0.5]M1.south) -- ([arrow-sh=-0.5]2*\dist, \logpos);
  \draw[arrow-notify-a=cold] ([arrow-sh=1.5]M1.south) -- ([arrow-sh=1.5]2*\dist, \logpos);
  \draw[arrow-notify-a=darkcola] ([arrow-sh=0.5]M1.south) -- ([arrow-sh=0.5]2*\dist, \logpos) node[label, xshift=-2pt] {Notify};

  \draw[arrow-notify-b=colb] ([arrow-sh=-1.5]M2.south) -- ([arrow-sh=-1.5]4*\dist, \logpos);
  \draw[arrow-notify-b=colc] ([arrow-sh=-0.5]M2.south) -- ([arrow-sh=-0.5]4*\dist, \logpos);
  \draw[arrow-notify-b=cold] ([arrow-sh=1.5]M2.south) -- ([arrow-sh=1.5]4*\dist, \logpos);
  \draw[arrow-notify-b=darkcola] ([arrow-sh=0.5]M2.south) -- ([arrow-sh=0.5]4*\dist, \logpos) node[label, xshift=-2pt] {Notify};

  \node[replica, fill=colb] (B) at (2*\dist, \logpos-\logh-0.9) {\nb};
  \node[replica, fill=colc] (C) at (3*\dist, \logpos-\logh-0.9) {\nc};
  \node[replica, fill=cold] (D) at (4*\dist, \logpos-\logh-0.9) {\nd};

  \draw[arrow-record-a=colb] ([arrow-sh=0.5]B.north) -- ([arrow-sh=0.5]2*\dist, \logpos-\logh);
  \draw[arrow-record-a=colc] ([arrow-sh=0.5]C.north) -- ([arrow-sh=0.5]3*\dist, \logpos-\logh);
  \draw[arrow-record-a=cold] ([arrow-sh=0.5]D.north) -- ([arrow-sh=0.5]4*\dist, \logpos-\logh);
  \draw[arrow-record-b=colb] ([arrow-sh=-0.5]B.north) -- ([arrow-sh=-0.5]2*\dist, \logpos-\logh) node[label, xshift=2pt] {Record};
  \draw[arrow-record-b=colc] ([arrow-sh=-0.5]C.north) -- ([arrow-sh=-0.5]3*\dist, \logpos-\logh) node[label, xshift=2pt] {Record};
  \draw[arrow-record-b=cold] ([arrow-sh=-0.5]D.north) -- ([arrow-sh=-0.5]4*\dist, \logpos-\logh) node[label, xshift=2pt] {Record};

  \draw[thick] (0.5, \logpos) -- (\logw, \logpos) node[midway, yshift=1pt, below] {Append-only Log};
  \draw[thick] (0.5, \logpos-\logh) -- (\logw, \logpos-\logh);

  \begin{scope}[shift={(8.25, \logpos-1.5*\logh)}]
    \coordinate (legend-pos) at (0, 0);

    \node[legend] at ($(legend-pos) + (0.5, 0)$) {\textbf{Arrow Legend:}};
    \node[legend] at ($(legend-pos) - (0, \legendspacing)$) {Sensor \sens{1}};
    \draw[arrow-record-a=black] ($(legend-pos) - (0, \legendspacing)$) -- ++(0.5, 0);
    \node[legend] at ($(legend-pos) - (0, 2*\legendspacing)$) {Sensor \sens{2}};
    \draw[arrow-record-b=black] ($(legend-pos) - (0, 2*\legendspacing)$) -- ++(0.5, 0);
    \node[legend] at ($(legend-pos) - (0, 3*\legendspacing)$) {Local data};
    \draw[arrow-local] ($(legend-pos) - (0, 3*\legendspacing)$) -- ++(0.5, 0);
  \end{scope}
\end{tikzpicture}

%% file: images/star-sens-mons.tex
\def\smsize{0.6cm}
\def\smpos{1}
\def\rsize{1cm}
\def\dist{1.0}
\def\logpos{0}
\def\logh{0.5}
\def\logw{12.7}

\definecolor{darkcola}{RGB}{140, 140, 180} %
\definecolor{cola}{RGB}{230, 230, 250} %
\definecolor{colb}{RGB}{100, 200, 100} %
\definecolor{colc}{RGB}{200, 100, 100} %
\definecolor{cold}{RGB}{100, 180, 180} %

\tikzset{
  arrow-record-a/.style={->, thick, color=#1},
  arrow-record-b/.style={->, thick, >={latex'}, color=#1},
  arrow-notify-a/.style={<-, thick, color=#1},
  arrow-notify-b/.style={<-, thick, >={latex'}, color=#1},
  arrow-local/.style={->, thick, dashed},
  arrow-sh/.style={xshift=#1*\arrowspacing},
}

\tikzset{
  label/.style={midway, yshift=-2pt, font=\scriptsize, fill=white, fill opacity=1, text opacity=1, inner sep=1.5pt},
  legend/.style={left, font=\scriptsize},
  box/.style={draw, rectangle, minimum width=\smsize, minimum height=\smsize},
  replica/.style={draw, rectangle, minimum width=\rsize, minimum height=\smsize},
  sensor-box/.style={box, text=black, preaction={fill=#1}},
  monitor-box/.style={box, text=black, preaction={fill=#1}},
  replica-box/.style={replica, preaction={fill=#1}},
  legend-box/.style={font=\scriptsize, draw, rectangle, minimum width=4cm, minimum height=0.35cm, text=black, preaction={fill=#1}, anchor=west, inner sep=0pt, align=right, text height=1.5ex, text depth=.25ex},
}

\begin{tikzpicture}

  \draw[thick, fill=cola] (0.5, 2.1) rectangle (\logw, \logpos+0.5) node[anchor=north east] at (\logw, 2.1) {Replica \na};

  \node[sensor-box=light-yellow]  (LS) at (1*\dist, \smpos)   {\textsf{LS}};
  \node[monitor-box=light-yellow] (LM) at (2*\dist, \smpos)   {\textsf{LM}};
  \node[sensor-box=light-blue]    (MS) at (3.5*\dist, \smpos) {\textsf{MS}};
  \node[monitor-box=light-blue]   (MM) at (4.5*\dist, \smpos) {\textsf{MM}};
  \node[sensor-box=light-orange]  (SS) at (6*\dist, \smpos)   {\textsf{SS}};
  \node[monitor-box=light-orange] (SM) at (7*\dist, \smpos)   {\textsf{SM}};
  \node[sensor-box=light-green]   (CS) at (8.5*\dist, \smpos) {\textsf{CS}};
  \node[monitor-box=light-green]  (CM) at (9.5*\dist, \smpos) {\textsf{CM}};
  \node[replica-box=light-blue]  (RSM) at (12*\dist, \smpos) {RSM};

  \draw[arrow-local] (LM.north) to[out=60, in=120, looseness=0.5] (SS.north);
  \draw[arrow-local] (LM.north) node[font=\scriptsize, above] {\lm} to[out=60, in=120, looseness=0.4] (CS.north);
  \draw[arrow-local] (MM.north) node[font=\scriptsize, above] {\Faulty} to[out=60, in=120, looseness=0.5] (SM.north);
  \draw[arrow-local] (SM.east) -- (CS.west) node[midway, yshift=-2pt, font=\scriptsize, above] {\Cand, \un};
  \draw[arrow-local] (CM.east) -- (RSM.west) node[midway, yshift=-3pt, font=\scriptsize, above] {Reconfigure};

  \draw[arrow-record-a=darkcola] (LS.south) -- (1*\dist, \logpos) node[label] {\latm{Prop}{A}{B}};
  \draw[arrow-record-a=darkcola] (MS.south) -- (3.5*\dist, \logpos) node[label] {\compm};
  \draw[arrow-record-a=darkcola] (SS.south) -- (6*\dist, \logpos) node[label] {\suspm{Slow}{A}{B}};
  \draw[arrow-record-a=darkcola] (CS.south) -- (8.5*\dist, \logpos) node[label] {\cfgm};

  \begin{scope}[shift={(\logw+0.5, \logpos+0.5)}]
    \node[legend-box=light-yellow] at (0, 1.5) {\textsf{LS/LM: Latency Sensor/Monitor}};
    \node[legend-box=light-blue]   at (0, 1.0) {\textsf{MS/MM: Misbehavior Sensor/Monitor}};
    \node[legend-box=light-orange] at (0, 0.5) {\textsf{SS/SM: Suspicion Sensor/Monitor}};
    \node[legend-box=light-green]  at (0, 0.0) {\textsf{CS/CM: Config Sensor/Monitor}};
  \end{scope}
\end{tikzpicture}

%% file: tex/monitor_latency.tex
\subsection{Low-Latency Role Assignment}
\label{sec:optilog-sensors}

In this section, we present a pipeline of four sensors and monitors suitable for assigning roles for low-latency configurations under Byzantine failures.
The pipeline comprises:
a \latsensor to measure link latencies,
a \sussensor to suspect replicas whose protocol messages are significantly delayed compared to their reported latencies,
a \missensor to report on provable misbehavior, and
a \cfgsensor to search for an efficient configuration.

While Aware~\cite{aware} uses latency measurements to optimize configurations, our novelty is in combining latency measurements with suspicions, both to detect protocol latencies that deviate from expectations and to incorporate suspicions into the configuration search.
Additionally, through the use of the generic sensor and monitor abstraction, we can design reusable components, as we show by integrating them into different BFT systems in \cref{sec:oware} and \cref{sec:optitree}.
Applying our sensors to different BFT protocols and topologies requires certain functionality to be provided by the underlying protocol, which we highlight in the description of each sensor.

\subsubsection{Latency Monitoring}
\label{sec:latency}

We now discuss our approach to measuring latencies between replicas.
For protocols with direct message exchanges, like HotStuff~\cite{hotstuff}, the latency sensor measures the \textit{round-trip latency} of the protocol messages themselves.
For protocols where replicas do not issue direct replies, the sensor instead relies on dedicated probe messages.
Each replica compiles its measurements into a \textit{latency vector}, which it submits to the \consmod.
Any replica that fails to reply is marked as $\infty$ in the latency vector.
We assume that the \consmod prevents the censoring of latency vectors---for example, by forcing a leader change upon message omission~\cite{pbft} or rotating leaders frequently~\cite{hotstuff}.

Each replica's \latmonitor maintains a \textit{latency matrix}~\lm, representing inter-replica latencies.
When a measurement is committed, the monitor updates the sender's row with the new latency vector.
Symmetry is preserved by setting $\lme{\na}{\nb}=\lme{\nb}{\na}=\max(\lat[r]{\na}{\nb}, \lat[r]{\nb}{\na})$, where $\lat[r]{\na}{\nb}$ is the \textit{recorded} latency between $\na$ and $\nb$.

While measuring latencies of protocol messages may reduce overhead and yield precise estimates, probe messages can fill the latency matrix faster.
In HotStuff, for instance, a rotating leader needs $2\n$ views to fill the matrix when it relies solely on protocol messages.
While measurements using protocol messages must be implemented within the protocol, the \latsensor, which uses probe messages, does not require any functionality from the underlying protocol.

We measure latency periodically.
When measurement is performed after GST, the recorded latency between correct replicas is within a factor $\delta$ of the actual latency.
No assumption is made about recorded latencies between replicas where at least one replica is faulty.
However, our suspicion sensor, presented below, is designed to detect faulty replicas that under-report their latencies.
Moreover, in the examples we consider, faulty replicas that inflate the latencies of adjacent links are less likely to be assigned critical roles.

%% file: tex/monitor_misbehavior.tex
\subsubsection{Misbehavior Monitoring}
\label{sec:mms}

A significant challenge for scalable RSMs is to stay optimized despite faulty replicas disrupting the configuration.
\sysname provides two monitors for detecting faulty replicas: one precise (this section) and one for suspicions (\cref{sec:suspicion}).

Precise detection is possible using the \textit{proof-of-misbehavior} technique~\cite{ia-ccf, BFTFD, byzid,zyzzyva, prime}.
The \missensor, integrated into the \consmod, observes the replicas' protocol-specific behavior and raises a \textit{complaint} when detecting a provable violation.
Complaints are signed and proposed via the \sensapp and logged.
Detected misbehaviors include invalid threshold signatures, proposals, votes, complaints, and equivocation.
An equivocation is identified when replicas detect the leader sending different messages despite the protocol dictating it should send identical messages.
What entails a misbehavior depends on the protocol or even the implementation. Therefore, the detection functionality for the \missensor must be provided by the underlying protocol.

When a complaint proposal is received, each replica's \mismonitor verifies it.
If valid, the monitor updates its \textit{list of provably faulty replicas} \Faulty, which is used to exclude replicas from special protocol roles, such as leader.

%% file: tex/monitor_suspicion.tex
\subsubsection{Suspicion Monitoring}
\label{sec:suspicion}

Gathering enough evidence for a proof-of-misbehavior complaint is often infeasible.
We therefore add suspicion monitoring to detect timing and omission faults.
These suspicions are then used to build a candidate set that excludes the suspected replicas.

Our \sussensor applies to leader-driven RSM protocols where proposals are committed through repeated message exchanges.
We refer to each proposal-commit cycle as a round.
The \sussensor assumes replicas have access to synchronized clocks.

The \sussensor's main goal is to detect replicas that fail to meet their claimed latencies as reported in \lm.
Detecting these deviations is challenging, since protocol messages often do not follow a request-reply pattern.
For example, consider the message pattern in PBFT~\cite{pbft}, illustrated in \cref{fig:pbft-times}.
For replica $\na$ to determine whether the highlighted message $m$ is delayed, it must know when replica $\nb$ should have sent $m$.

To solve this issue, we let the leader timestamp its proposal when starting a round and use this timestamp as a reference point for timeouts on individual messages.
Thus, in our example, replica $\na$ would expect message $m$ to arrive a certain time after the proposal timestamp.
This timeout is shown as $\dm$ in \cref{fig:pbft-times}.

\begin{figure}[t]
  \input{images/pbftdelta.tex}
  \caption{PBFT message pattern showing expected round duration $\rdur$ and message delay $\dm$.}
  \label{fig:pbft-times}
\end{figure}

Additionally, the \sussensor checks that leaders do not delay timestamping.
It does so by verifying that consecutive proposal timestamps fall within the expected round duration, denoted $\rdur$.

The \sussensor requires the underlying protocol to provide functionality to compute the estimated round duration $\rdur$ for a configuration.
The protocol must also provide $\dm$ for each message.

The conditions for raising a suspicion are summarized in the table below.
Here, \susp{\na}{\nb} denotes that $\na$ suspects $\nb$.

\begin{table}[ht]
  \footnotesize
  \centering
  \begin{tabular}{@{}p{0.1cm}p{5.25cm}p{1.78cm}@{}}
    \toprule
    \multicolumn{2}{@{}l}{\textbf{Condition for Suspicion}}       & \textbf{Suspicion} \\ \midrule
    (a) & Proposal timestamps not within $\df{\rdur}$ & \suspm{Slow}{\na}{\nl} \\
    (b) & No message $m$ from $\nb$, $\df{\dm}$ after round start & \suspm{Slow}{\na}{\nb} \\
    (c) & False suspicion from $\nb$ on $\na$ & \suspm{False}{\na}{\nb} \\
    \bottomrule
  \end{tabular}
\end{table}

\noindent
First, \susp{\na}{\nl} if two consecutive proposal timestamps from the leader differ by more than $\df{\rdur}$.
As discussed above, this ensures that rounds are started at the correct interval.
Second, \susp{\na}{\nb} if a message $m$ from $\nb$ to $\na$ does not arrive within $\df{\dm}$ after the leader's proposal timestamp.
This mechanism also raises suspicions against the leader if it creates the proposal on time but delays sending it, thereby delaying round initialization.
The timeouts $\dm$ and $\rdur$ are adjusted by the factor $\delta$ to account for variations in latency.
Finally, if a \suspm{\_}{\nb}{\na} is raised, then $\na$ reciprocates by raising a \suspm{False}{\na}{\nb}.
The last condition helps to distinguish suspicions raised against crashed replicas, which cannot reciprocate, from suspicions against misbehaving replicas.

In complex message patterns, a single delayed message may delay subsequent messages and trigger multiple suspicions.
Such causally related suspicions are filtered in the \susmonitor.
To enable this filtering, a suspicion includes the message type and round that caused it.
For each round, the \susmonitor retains only the first suspicion, discarding suspicions from later protocol phases.
Additionally, if the leader raises suspicion for a message in round $i$, suspicions against a delayed proposal timestamp in round $i+1$ are also filtered.

The \sussensor guarantees that if a protocol does not finish a round within $\df{\rdur}$, a suspicion will be raised.
Additionally, if $\rdur$ and $\dm$ are configured correctly, then after \ac{GST}, any suspicion between correct processes will be filtered out.
Thus, the \sussensor ensures that, after \ac{GST}, a configuration either meets expected latencies or raises a useful suspicion.
Appendix~\ref{sec:suspicion_proof} formalizes the properties that $\rdur$ and $\dm$ must satisfy and provides proofs of these guarantees.

\paragraph{Candidate Selection}
The \susmonitor processes suspicions to produce a \textit{candidate set} \Cand of replicas considered correct, along with an estimate $\un$ of the number of misbehaving replicas.
To this end, the \susmonitor first removes provably faulty replicas \Faulty, reported by the \mismonitor.
Next, the \susmonitor distinguishes between crash suspicions and misbehavior suspicions to produce a good estimate $\un$.

The \susmonitor must ensure that faulty replicas causing suspicions are removed from the candidate set and that the number of misbehaving replicas is not underestimated.
Failure to do so may allow faulty replicas to repeatedly interrupt the system, by rendering a configuration dysfunctional.
On the other hand, the \susmonitor must provide a sufficiently large candidate set to ensure all critical roles can be assigned to candidates.
Further, a larger candidate set and a lower estimate of misbehaving replicas permit more choices when assigning roles and may thus result in a more efficient configuration.
We design the \susmonitor to always provide $n-f$ candidates and to ensure that if repeated suspicions are raised against faulty replicas after GST, eventually none of these replicas are included in the candidate set, as we prove in Appendix~\ref{sec:suspicion_proof}.

Condition (c) allows us to distinguish suspicions against crashed replicas from those caused by misbehaving replicas.
A crashed replica $\nb$ will be suspected by others through $\suspm{Slow}{\na}{\nb}$.
If a faulty replica $\nb$ raises $\suspm{Slow}{\nb}{\na}$ against a correct replica $\na$, then $\na$ will reciprocate with $\suspm{False}{\na}{\nb}$.
Thus, a replica suspected of being slow without raising a counter-suspicion is \textit{considered crashed}.
Conversely, if two replicas suspect each other, it remains unclear which one is faulty.
The \susmonitor keeps separate data structures for these two cases: a set $\Crash$ for crashed replicas and a \textit{suspicion graph}~$\G$ for other suspicions.
Formally, $\G=(\V,\E)$ is an undirected graph where the vertices $\V=\{\Pi\setminus\Faulty\setminus\Crash\}$ represent possible candidate replicas, and an edge $(\na, \nb) \in \E$ indicates a two-way suspicion, $\susptwo{\na}{\nb}$.

When a suspicion $\suspm{Slow}{\na}{\nb}$ is raised against a correct replica $\nb$, it may take several views before a reciprocation $\suspm{False}{\nb}{\na}$ is logged, especially if faulty replicas attempt to censor such suspicions.
We therefore treat every new suspicion between two replicas in $\V$ as a two-way suspicion by adding the corresponding edge to $\G$.
If no reciprocation is received within $\f+1$ leader changes (views), the edge is treated as a one-way suspicion, and $\nb$ is added to $\Crash$.

The \susmonitor computes a maximum independent set from the vertices in $\G$ and uses it as the candidate set $\Cand$.
This computation must be deterministic, to ensure that all correct replicas reach the same conclusion.
Additionally, the monitor outputs the estimated number of misbehaving replicas, $\un=|\V|-|\Cand|$.

Before GST, unstable latencies may cause suspicions, even between correct replicas.
To prevent these suspicions from persisting indefinitely, the \susmonitor employs two mechanisms.
First, if the system remains stable with no new suspicions for \winlen views, the monitor begins removing old suspicions in the order they appear in the log.
Second, if too many suspicions are received, it also starts discarding old suspicions.
Too many suspicions occur when $\G$ no longer contains an independent set of size $\n-\f$.

%% file: images/pbftdelta.tex
\centering
\begin{tikzpicture}
    \footnotesize
    \def\A{0}
    \def\B{0.5}
    \def\C{1.0}
    \def\L{1.5}

    \draw (0,\A) node[left]{$A$}[-latex] -- (5,\A);
    \draw (0,\B) node[left]{$B$}[-latex] -- (5,\B);
    \draw (0,\C) node[left]{$C$}[-latex] -- (5,\C);
    \draw (0,\L) node[left]{$L$}[-latex] -- (5,\L);

    \draw (0.25,\L)[thick,-latex,gray] -- (0.75,\C);
    \draw (0.25,\L)[thick,-latex,gray] -- (0.75,\B);
    \draw (0.25,\L)[thick,-latex,gray] -- (0.75,\A);

    \foreach \f/\t in { \C/\A, \B/\A,\C/\L, \C/\B, \B/\L, \A/\C, \A/\L,  \B/\C, \A/\B} {
        \draw (1,\f) [thick,-latex,gray] -- (1.5,\t);
    }
    \foreach \f/\t in { \C/\L, \C/\B, \C/\A, \B/\L, \B/\C, \B/\A, \A/\L, \A/\C, \A/\B} {
        \draw (1.75,\f) [thick,-latex,gray] -- (2.25,\t);
    }

    \draw (0.3, \A-0.1)[|-|] -- node[below]{$\rdur$}  (2.25,\A-0.1);

    \draw (2.5,\L)[thick,-latex,gray] -- (3,\C);
    \draw (2.5,\L)[thick,-latex,gray] -- (3,\B);
    \draw (2.5,\L)[thick,-latex,gray] -- (3,\A);

    \foreach \f/\t in { \C/\A, \B/\A,\C/\L, \C/\B, \B/\L, \A/\C, \A/\L,  \B/\C, \A/\B} {
        \draw (3.25,\f) [thick,-latex,gray] -- (3.75,\t);
    }
    \foreach \f/\t in { \C/\L, \C/\B, \C/\A, \B/\L, \B/\C, \B/\A, \A/\L, \A/\C, \A/\B} {
        \draw (4,\f) [thick,-latex,gray] -- (4.5,\t);
    }
    \draw (3.25,\B) [thick,-latex] -- node[right,xshift=0.1cm]{$m$} (3.75,\A);
    \draw (2.5, \A-0.1)[|-|] -- node[below]{$\dm$}  (3.75,\A-0.1);
\end{tikzpicture}

%% file: tex/monitor_config.tex
\subsubsection{Configuration Monitoring}
\label{sec:config_sensor}
\Sysname's monitoring architecture is also used to trigger reconfigurations.
The \cfgsensor searches for a new configuration and proposes it via the log.
The \cfgmonitor then selects among the proposed configurations and initiates reconfiguration.

This separation has several benefits.
First, while the monitor's reconfiguration decision must be deterministic and taken by all replicas, the sensor's search for a better configuration can be performed by a small subset of replicas.
Second, when the search space is large, the sensor can use more efficient non-deterministic search algorithms.
Lastly, the search space can be divided among the replicas.

The \cfgsensor utilizes measurements from the other \monitors to find a new configuration.
We say that a configuration is valid if all special roles are assigned to replicas from the candidate set $\Cand$.
A search can be triggered periodically, or by an updated $\Cand$ that invalidates the current configuration.
The remaining configurations are ranked using a \scorefn function, and the best-scoring configuration, \Cfg, is proposed to the log.
To achieve low-latency configurations, \scorefn predicts the round duration $\rdur$ using the latency matrix $\lm$ from the \latmonitor together with the estimate of misbehaving replicas $\un$ from the \susmonitor.
By incorporating $\un$, \scorefn can account for potential message delays caused by misbehaving replicas.
Thus, \scorefn can assume that $\un$ replicas with normal roles will not respond.
Hence, by using $\un$, \scorefn can adapt to the actual number of faults in the system, yielding significantly faster configurations compared to assuming the worst-case of $\f$ faults.

When assigning multiple roles, it quickly becomes intractable to compute the score for all configurations.
In such cases, we use a heuristic optimization technique based on simulated annealing~\cite{simulatedannealing}.
We initiate the search with a random configuration, ensuring that special roles are only selected from \Cand.
Simulated annealing uses a \textit{mutate} function to randomly swap two replicas in the configuration.
Our mutate function ensures that replicas with special roles are only swapped with other replicas from \Cand.

Simulated annealing then compares the score of the mutated configuration with that of the best configuration found so far and probabilistically selects the new configuration for further mutation.
While simulated annealing allows exploration beyond local minima, it does not guarantee a global minimum.
The search ends when the \textit{search timer} expires or when simulated annealing converges.
At this point, the \cfgsensor logs the best configuration found.

We note that our search procedure is not deterministic. Indeed, when different replicas start the search with different initial configurations, it is more likely that one of them finds the best configuration.

When the \cfgmonitor receives a configuration proposal \Cfg, it first checks that the proposal represents a \textit{valid configuration}.
The monitor also checks if the current configuration has become invalid, e.g., due to an updated \Cand from the \susmonitor.
If the current configuration is invalid, reconfiguration is needed.
However, the monitor waits for at least $\f+1$ replicas to propose new configurations before selecting the one with the best score.
This prevents reconfiguring to a potentially suboptimal configuration proposed by a faulty replica.
The monitor can trigger a reconfiguration even if the current configuration is still valid but has been active for a while.
In this case, the monitor requires the new configuration to have a significantly better score to avoid frequent reconfigurations that could disrupt the system.
We note that a replica can validate the score of configurations proposed by others and independently decide which configuration to accept only because measurements are recorded in the log and metrics are consistent across replicas.

%% file: tex/aware.tex
\section{\oware}
\label{sec:oware}
We first study how \sysname{} can be applied to BFT-SMaRt~\cite{bft-smart}, a modern PBFT-like replication library that features many optimizations.
These optimizations are often realized through protocol extensions such as Wheat~\cite{wheat} and Aware~\citep{aware}, which we use for our study.
Both protocols assign higher voting weights to some replicas.
These weighted votes enable faster consensus, especially when $\n>3\f+1$, as fewer replies are needed to make progress in favorable conditions.

Aware extends Wheat with infrastructure for measuring round-trip latencies using probe messages.
It also adds a reconfiguration framework that reassigns the leader role and weights to minimize a configuration's latency.
However, Aware does not include infrastructure to detect misbehaving replicas, for example replicas that respond quickly to probes but handle protocol messages slowly.

\oware applies \sysname to Aware.
\oware's main novelty is the addition of misbehavior and suspicion monitoring.
Further, when searching for a configuration, \oware avoids replicas not in the candidate set.

As mentioned in \cref{sec:suspicion} and \cref{sec:config_sensor}, the main functionality necessary to apply \sysname to a protocol are
\begin{enumerate}
\item provide a \scorefn function for configuration search
\item provide a procedure to estimate \rdur ~and $\dm$.

\end{enumerate}
Aware already provides a \scorefn function that computes the round duration \rdur ~from the latency matrix~\lm.
We use that function in \oware.
The \scorefn function computes the estimated time for the arrival of messages in a round, ignoring the slowest messages, when a quorum is formed.
For the duration estimate $\dm$ we use the same approach as \scorefn, but only consider messages causally preceding the sending of a message $m$ in a round.

By integrating \sysname{} into Aware, we can achieve an adaptive BFT protocol that can optimize its configuration to achieve low latency in a WAN (which is what Aware does) \textit{and} react to malicious delays and omissions carried out by Byzantine replicas that try to slow down the system.

%% file: tex/optitree.tex
\section{\optitree}
\label{sec:optitree}

This section introduces \optitree, a variant of Kauri~\cite{kauri}, utilizing \sysname to construct \validt, low-latency trees.

\subsection{Preliminaries for \optitree}
\label{sec:optitree-prelim}

\subsubsection{A Brief Primer on Kauri}
\label{sec:kauri}

This section gives an overview of Kauri and its reconfiguration approach.
To improve scalability, Kauri replaces HotStuff's star topology with a tree topology (\cref{fig:plain-tree}), where leader \nr disseminates proposals top-down and aggregates votes bottom-up.
We refer to replicas with a specific role or placement in the tree as nodes.
In this structure, Kauri's \textit{internal nodes}, $\Internal = \{\nr, \nin_1, \ldots, \nin_b\}$, manage only $\bn$ child nodes, where $\bn$ is the tree's branch factor.
As a result, the internal nodes experience a much lower load than HotStuff's leader, which must interact with $\n-1$ replicas.
Let $\Intermediate = \Internal \setminus \{\nr\}$ denote the \textit{intermediate nodes}, and $\Leaf = \C{} \setminus \Internal = \{\nleaf{1}, \ldots, \nleaf{b^2}\}$ the \textit{leaf nodes}.

\begin{figure}[hbt]
  \centering
  \input{images/plain-tree}
  \caption{Tree with $\n=13$ replicas and branch factor $\bn=3$.}
  \Description{Diagram showing a tree with 13 replicas and a branch factor of 3.}
  \label{fig:plain-tree}
\end{figure}

Kauri exploits the parallelism of its tree topology through pipelining.
Specifically, the root can initiate multiple consensus instances concurrently, resulting in higher throughput than a comparable star topology.

However, it is hard to construct trees in the presence of faults.
That is, the internal nodes of a tree can control both the dissemination and aggregation stages of the protocol.
In a star topology, assuming an honest leader, protocol execution may succeed despite some faulty votes.
This may not be the case for a tree topology.
To address this concern, Kauri forms a \validt tree using a reconfiguration algorithm based on \tbc~\cite{kauri}.
This method involves dividing the replicas into $\tn$ disjoint bins, assuming that $\f<\tn$, where $\tn=\n/\inn$ and $\inn$ is the number of internal nodes.
Kauri then constructs trees using these bins, with replicas in each bin serving as internal nodes and the rest as leaf nodes.
If $\f<\tn$, there is always at least one fault-free bin, ensuring at least one tree with correct internal nodes.
If Kauri is unable to establish a \validt tree within $\tn$ reconfiguration trials, it reverts to using a star topology.
Neiheiser et al.~\cite{kauri} show that a three-level tree can guarantee liveness if the internal nodes are correct.
Their reconfiguration algorithm exploits this guarantee by iterating through trees constructed from disjoint sets of internal nodes.

\subsubsection{Kauri's Challenges}

In a wide-area environment, Kauri may face several challenges:
1)~A single faulty replica may exclude many other replicas from being considered as internal nodes.
Moreover, if latencies are not uniform, the performance of successive trees may degrade significantly.
We show this in \cref{sec:security-latency}.
2)~Kauri reconfigures only when a tree fails to collect $\q$ replies.
For example, if a low-latency subtree fails, Kauri may wait for a higher-latency subtree to respond, resulting in degraded performance.
3)~Kauri's reconfiguration approach only supports $\mathcal{O}(\sqn)$ reconfigurations.

To address Challenge~1, \optitree uses the suspicion sensor of \sysname to eliminate faulty replicas from the candidate set individually or paired with at most one correct replica.
For Challenge~2, the estimate on faulty nodes $\un$ provided by the \susmonitor in \sysname allows \optitree to successively adjust the failure threshold of the tree.
Further, we prove that in any scenario with up to $\f$ faults, \optitree uses at most $2\f$ reconfigurations to find a \validt tree, addressing Challenge~3.

\subsection{Overview of \optitree}

\optitree finds \validt, low-latency trees for large-scale RSM deployments.
A \textit{\validt tree} is one where the internal nodes are correct, ensuring that the tree can collect a quorum of votes.
\optitree assigns the roles of root, intermediate node, and leaf node connected to a specific internal node.
As observed in Kauri~\cite{kauri}, a tree remains functional as long as all internal nodes (root and intermediate) are correct.
We therefore only consider the roles of internal nodes as special, ensuring that \sysname assigns these roles to candidate replicas considered correct (\Cand).

However, an arbitrary \validt tree may be suboptimal, and so \optitree additionally aims to minimize the tree's latency based on the replicas' recorded latencies.
We define \textit{tree latency} as the time it takes for the root to disseminate the proposal and aggregate a quorum of votes.
A tree fails when it cannot collect a quorum of votes within a timeout.
In the following, we explain how \sysname's sensors and monitors are adjusted to match the characteristics of a tree topology.

In \cref{sec:ol-trees}, we show how to provide the necessary functionality to apply \sysname to Kauri.
In \cref{sec:tree-suspicion}, we describe a new variant of the \susmonitor that finds a working tree in fewer iterations.
This shows how the \sysname architecture can further be adjusted to a specific use case.

%% file: images/plain-tree.tex
\def\levelsep{.7} %
\def\nodesep{.7}  %

\footnotesize
\begin{tikzpicture}
    [vertex/.style={circle, draw=black, fill=black!10!white, thick, inner sep=0pt, minimum size=4mm, align=center},
        new/.style={draw=red, fill=red!30!white},
        exclude/.style={fill=red!40!white},
        tree/.style={draw=blue!60!white, thick},
        label/.style={right=4pt}]

\coordinate (a) at (4 * \nodesep, \levelsep);
\coordinate (b) at (1 * \nodesep, 0);
\coordinate (c) at (4 * \nodesep, 0);
\coordinate (d) at (7 * \nodesep, 0);
\coordinate (e) at (0 * \nodesep, -\levelsep);
\coordinate (f) at (1 * \nodesep, -\levelsep);
\coordinate (g) at (2 * \nodesep, -\levelsep);
\coordinate (h) at (3 * \nodesep, -\levelsep);
\coordinate (i) at (4 * \nodesep, -\levelsep);
\coordinate (j) at (5 * \nodesep, -\levelsep);
\coordinate (k) at (6 * \nodesep, -\levelsep);
\coordinate (l) at (7 * \nodesep, -\levelsep);
\coordinate (m) at (8 * \nodesep, -\levelsep);

\draw[tree] (a)--(b);
\draw[tree] (a)--(d);
\draw[tree] (b)--(e);
\draw[tree] (b)--(f);
\draw[tree] (d)--(k);
\draw[tree] (d)--(m);
\draw[tree] (a)--(c);
\draw[tree] (b)--(g);
\draw[tree] (c)--(h);
\draw[tree] (c)--(j);
\draw[tree] (c)--(i);
\draw[tree] (d)--(l);

\node[vertex] at (a) {$\nr$};
\node[vertex] at (b) {$\nin_1$};
\node[vertex] at (c) {$\nin_2$};
\node[vertex] at (d) {$\nin_3$};
\node[vertex] at (e) {$\nleaf{1}$};
\node[vertex] at (f) {$\nleaf{2}$};
\node[vertex] at (g) {$\nleaf{3}$};
\node[vertex] at (h) {$\nleaf{4}$};
\node[vertex] at (i) {$\nleaf{5}$};
\node[vertex] at (j) {$\nleaf{6}$};
\node[vertex] at (k) {$\nleaf{7}$};
\node[vertex] at (l) {$\nleaf{8}$};
\node[vertex] at (m) {$\nleaf{9}$};

\end{tikzpicture}

%% file: tex/optitree_monitors.tex
\subsection{\sysname for Tree Topologies}
\label{sec:ol-trees}

To apply \sysname to assign roles for a tree topology, we must implement the functionality highlighted in \cref{sec:optilog-sensors}, namely a \scorefn function for configuration search and a procedure to estimate \rdur~and $\dm$.
Further, for \optitree we also augment the \mismonitor with a special rule, as described below.

Our \scorefn function is defined in \cref{def:score}.
It computes the minimum latency required to collect votes from $\kn=\q+\un$ nodes in a tree, where $\un$ is the estimate of faulty nodes provided by the \susmonitor.
This estimate ensures that the tree can meet the predicted latency despite up to $\un$ unresponsive nodes.

\input{tex/optitree_ranking}

The \scorefn function is also used to estimate the round duration $\rdur$.
For message duration estimate $\dm$ we sum the latency of causally preceding messages in a round.

\optitree also provides an additional rule for misbehavior monitoring of invalid vote aggregation.
In a tree configuration, intermediate nodes are responsible for aggregating votes from their children and forwarding them to the root.
However, if an intermediate node does not receive a vote from a child node, the aggregate vote must include a suspicion for the missing vote.
That is, the aggregate must include $\bn + 1$ votes or suspicions.
Failure to meet this requirement allows the root to use the aggregation vote as \textit{proof-of-misbehavior} against the intermediate node.

We note that, as an additional optimization, the suspicion sensor can be simplified in \optitree.
For example, for suspecting leaf nodes, intermediate nodes can rely on measuring round trip times, rather than relying on the delay relative to proposal timestamps.
Additionally, due to the added misbehavior rule, intermediate nodes need to either include a leaf's vote on time or suspect that leaf.
Therefore, suspicion monitoring for condition (b) in \cref{sec:suspicion} can be omitted at leaf nodes.

\input{tex/optitree_suspicion}

%% file: tex/optitree_ranking.tex
\begin{definition}\label{def:score}
  For a tree $\tree$, the $\score(\kn, \tree)$ is the minimum latency required to collect votes from $\kn=q+\un$ nodes.

  This score is computed using the link latencies in the tree, derived from $\lm$, taking into account that intermediate nodes first collect votes from their children before forwarding them to the root.

  The \textit{aggregation latency} $\agt{\nin}$, for intermediate node $\nin$, is the maximum latency from $\nin$ to any of its children, \Ch{\nin}:
  \[
    \agt{\nin} = \max_{\nv \in \Ch{\nin}} \lme{\nin}{\nv}
  \]
  Let $\Intermediate_k$ represent subsets of all intermediate nodes such that the subtrees rooted at these nodes contain a total of at least $k$ nodes:
  \[
    \Intermediate_k = \{M \subset \Intermediate \mid \sum_{\nin \in M} |\Ch{\nin}| + 1 \geq k \}
  \]
  Thus, the time to collect $\kn$ votes is the minimum time to collect aggregates from subtrees in a set from $\Intermediate_{\kn-1}$, since the root $\nr$'s vote is added separately.
  \[
    \score(\kn, \tree) = \min_{M \in \Intermediate_{\kn-1}} \left( \max_{\nin \in M} \left( \agt{\nin} + \lme{\nin}{\nr} \right) \right)
  \]
\end{definition}

%% file: tex/optitree_suspicion.tex
\subsection{Candidates for Trees}
\label{sec:tree-suspicion}

\optitree assigns the roles of internal nodes to replicas in the candidate set \Cand.
However, since a configuration has only $b+1\approx \sqrt{n}$ many internal nodes, \optitree does not require that the candidate set contains $\n-\f$ replicas, as is the case in \sysname.
In the following, we describe an updated \susmonitor that computes a smaller candidate set but guarantees that the replicas in that candidate set are correct after fewer reconfigurations (at most $2\f$) compared to the approach in \cref{sec:suspicion}.
This shows the flexibility of \sysname's architecture.

When computing a maximum independent set for candidate selection (\cref{sec:suspicion}), a single intermediate node suspected by a leaf may cause the candidate set to update and require a reconfiguration.
Faulty replicas may force $\Omega(\f^2)$ reconfigurations before they are excluded from the candidate set~\cite{jehl2019quorum}.

We introduce two additional sets, \MG and \T, derived from the suspicion graph $\G$.
\MG is a maximal set of disjoint edges in $\G$.
For every edge in \MG, at least one of the adjacent vertices is a faulty replica.
We therefore exclude both replicas from $\Cand$.
Whenever an edge is added to \G, the \susmonitor checks if \MG is still maximal.
This check includes possibly removing one edge from \MG and adding two new ones.
The cost of maintaining \MG is $O(e^2)$, where $e$ is the number of edges in \G.
We also determine $\T$, the set of vertices that are not adjacent to an edge in \MG but part of a triangle in $\G$ with an edge in \MG.
We also exclude replicas in $\T$ from $\Cand$.

Finally, we return $\Cand$ containing replicas (vertices) in $\G$ that are not in $\T$ or adjacent to an edge in \MG.
We also return the estimated number of faulty replicas $\un=|\MG|+|\T|$.

\input{images/treeexample.tex}

\cref{fig:suspicionsexample} illustrates a suspicion graph $\G$ and the resulting tree configuration with excluded nodes.
In this example, $\MG = \{(\susn{1},\susn{4}), (\susn{2},\susn{3})\}$, as adding any other suspicion would violate the disjoint property of \MG, which is already maximal.
The set $\T = \{\naT\}$ includes nodes that form a triangle with one edge in $\MG$ and two others in $\G$.
As $\nbC$ has a one-way suspicion, it is added to $\Crash$.
Since all replicas from $\T$, $\MG$, and $\Crash$ are excluded from $\C{}$, we get $\Cand = \{\cn{1},\cn{2},\cn{3},\nr\}$.

We provide proofs of \optitree's correctness in Appendix~\ref{sec:security_analysis}.
We prove that after GST, at most $2\f$ reconfigurations are needed to form a correct tree.
Especially, if $\tn<\f$ replicas are faulty, at most $2\tn$ reconfigurations are needed.

%% file: images/treeexample.tex
\begin{figure}
  \centering\
  \footnotesize
  \begin{tikzpicture}
    [vertex/.style={circle, draw=black, fill=black!10!white,thick, inner sep=0pt, minimum size=3.5mm},
      new/.style={draw=red, fill=red!30!white},
      exclude/.style={fill=red!40!white},
      tree/.style={draw=blue!60!white, thick},
    label/.style={right=4pt}]
    \node[vertex, exclude] (A) at (0.5,0){\susn{1}};
    \node[vertex, exclude] (B) at (1.5,0){\susn{2}};
    \node[vertex, exclude] (C) at (2.5,0){\susn{3}};
    \node[vertex, exclude] (D) at (0,1){\susn{4}};
    \node[vertex, exclude] (E) at (1,1){$\naT$};
    \node[vertex] (F) at (2,1){\cn{1}};
    \node[vertex] (G) at (3,1){\cn{2}};
    \node[vertex, exclude] (H) at (0.5,2){$\nbC$};
    \node[vertex] (I) at (1.5,2){\cn{3}};
    \node[vertex] (J) at (2.5,2){$\nr$};

    \foreach \p/\c in {J/I, J/G, J/F, I/E, I/D, F/H, F/C, G/B, G/A} {
      \draw[tree] (\p) -- (\c);
    }

    \draw[thick, new] (A) -- (D);
    \draw[thick, new] (B) -- (C);
    \draw[dashed, thick, <->] (A) -- (D);
    \draw[dashed, thick, <->] (A) -- (E);
    \draw[dashed, thick, <->] (D) -- (E);
    \draw[dashed, thick, <->] (B) -- (C);
    \draw[dashed, thick, <->] (C) -- (G);
    \draw[dotted, thick, ->] (I) -- (H);

    \node[vertex,exclude] (X2) at (3.5,0){};
    \node[label] at (X2) {Excluded from internal nodes};

    \node[vertex] (X) at (3.5,0.5){};
    \node[vertex,exclude] (Y) at (4.5,0.5){};
    \draw[dotted, thick, ->] (X) -- (Y);
    \node[label] at (Y) {One-way suspicion};

    \node[vertex] (XX) at (3.5,1){};
    \node[vertex] (YY) at (4.5,1){};
    \draw[dashed, thick, <->] (XX) -- (YY);
    \node[label] at (YY) {Two-way suspicion};

    \node[vertex,exclude] (X3) at (3.5,1.5){};
    \node[vertex,exclude] (Y3) at (4.5,1.5){};
    \draw[thick, new] (X3) -- (Y3);
    \draw[dashed, thick, <->] (X3) -- (Y3);
    \node[label] at (Y3) {Two-way suspicion in $\MG$};

    \node[vertex] (X4) at (3.5,2){};
    \node[vertex] (Y4) at (4.5,2){};
    \draw[tree] (X4) -- (Y4);
    \node[label] at (Y4) {Edge in tree};

  \end{tikzpicture}

  \caption{Example of suspicions leading to removal of nodes as candidates for internal tree positions.}
  \Description{Graph depicting an example of suspicions leading to removal of nodes as candidates for internal tree positions.}
  \label{fig:suspicionsexample}
\end{figure}

%% file: tex/evaluation_intro.tex
\section{Evaluation}
\label{sec:evaluation}

In this section, we demonstrate \sysname's usefulness for BFT protocols by evaluating the use cases: \oware and \optitree.
For \oware, we show that using \sysname, particularly its suspicion monitoring, prevents Aware from performance degradation, for which we evaluate its runtime behavior under attack (§\ref{eval:oware:underAttack}).
We also assess the scalability of handling suspicions in \oware (§\ref{eval:oware:scalability}).

For \optitree, we show that integrating \sysname sensors and monitors yields better tree configurations, improving overall system performance.
After describing \optitree's implementation and experimental setup (§\ref{eval:optitree:setup}), we evaluate its performance against Kauri and HotStuff as baselines (§\ref{eval:optitree:baseline}).
We then examine the robustness of \optitree configurations under malicious replicas (§\ref{sec:security-latency}, §\ref{eval:optitree:tolerance}) and analyze how simulated annealing search time impacts overall tree latency (§\ref{eval:optitree:sa}).
Finally, we assess the overhead introduced by \sysname~(§\ref{sect:overhead}).

Artifacts used in the evaluation of \optitree and \oware, including experimental procedures and execution scripts, are available at \url{https://zenodo.org/records/17169372}.

\input{tex/evaluation_opt_aware}

\subsection{\optitree Implementation \& Experimental Setup}
\label{eval:optitree:setup}
We implemented \sysname and \optitree in an open-source BFT framework based on HotStuff~\cite{relab-hotstuff}.
We also implemented Kauri~\cite{kauri} in the same framework for comparison.

Our primary focus is evaluating our optimizations' performance in deployment scenarios where replicas are distributed globally.
To emulate such a distribution of replicas, we integrated a module mimicking network latency between replicas.
Each replica is assigned a \itz{location} from a predefined set of cities, and messages are artificially delayed based on the recipient's location.
Our network emulator includes 220 worldwide locations, obtained from WonderProxy~\cite{wonderproxy}, with intercontinental delays ranging from 150 to 250~ms, in addition to the actual network delay of 1~ms.

Our deployment cluster consists of 30 machines, each with 32GB of RAM and 12-core Intel Xeon processors, connected via a 10Gbps TOR switch.
We deployed 4 replicas per machine for all evaluations.
We used 4 clients to send requests to replicas, and replicas batch these requests into block proposals.
We do not emulate latency between clients and replicas; their communication is instantaneous.
Several existing studies~\cite{fluidity,aware,archer} examine client placement and its impact on BFT protocol performance.
We did not address this aspect, as it is orthogonal to our work.

We use the Chained HotStuff~\cite{hotstuff} protocol for consensus.
In all experiments, a fixed leader sends proposals and collects a quorum of votes to commit the proposal.
We use blocks of 1000 proposals, each without transaction payload.
\textit{Throughput} is the number of requests that replicas commit, while \textit{consensus latency} is the time between sending a block proposal and committing the block.
Unless specified, any reference to latency refers to consensus latency.

In Kauri, the leader (root) sends proposals in a tree structure, and responses are aggregated up the tree.
We implemented pipelining in Kauri by running multiple protocol instances concurrently.
We used 3 instances for all experiments with pipelining.
\optitree's main advantage is that it is able to form trees with lower latency.
Since pipelining has little effect on latency, we run most experiments without pipelining.
In all experiments, trees have a height of 3, and the configuration size \n determines the branching factor $\bn=(\sqrt{4\n-3}-1)/2$.
Taller trees tend to incur overly high latencies.
Trees are ranked using the $\score(\kn,\tree)$ function from \cref{def:score}, with $\kn=2\f+1$ as the default.

%% file: tex/evaluation_opt_aware.tex
\begin{figure}
	\centering
	\input{plots/opti-aware/exp1.tex}
	\input{plots/opti-aware/exp1aware.tex}
	\caption{Runtime attack: \textsc{\oware} vs. baselines.}
	\label{fig:optaware:exp1}
\end{figure}

\subsection{\oware Runtime Behavior}
\label{eval:oware:underAttack}
We implemented \oware inside the Aware framework~\cite{aware} by integrating \sysname's monitors. Since we now enriched  Aware with the capability to \textit{detect suspicions} (such as malicious replicas delaying messages), we are  interested in what happens when an attack occurs.
For this purpose, we evaluate \oware's runtime behavior under the known \textit{Pre-Prepare delay attack}~\cite{prime,clement2009making}, in which a Byzantine leader intentionally delays a proposal.
We are interested in its ability to detect delayed protocol messages through suspicions and reassign the leader role to maintain optimal performance.

\paragraph{Setup}
We deploy one client and one replica in each of 21 European cities using the Phantom network simulator~\cite{jansen2022co}.
Clients continuously issue requests to the replicas.
Our experiment compares the protocols BFT-SMaRt, Aware, and \oware.
We measure end-to-end request latency as observed by the clients.
During execution, a Byzantine leader launches a Pre-Prepare delay attack at approximately $t = 80\,$s to degrade system performance.

\paragraph{Observations}
\cref{fig:optaware:exp1} shows latency measurements from a representative client in Nuremberg.
Initially, for $t<40\,$s, all protocols exhibit comparable latency.
At $t = 40\,$s, Aware and \oware{} optimize their configuration (i.e., assigning the leader role and voting weights to replicas) as explained in \cref{sec:oware}.
This optimization yields a significant 35\% reduction in latency compared to BFT-SMaRt, which remains static.
The observed improvement aligns with prior findings on adaptive BFT protocols~\cite{aware,berger2024chasing}, which also demonstrate performance gains under different experimental setups.

At $t = 82\,$s, the Byzantine leader initiates the Pre-Prepare delay attack, delaying proposals to inflate client-observed latency.
All protocols experience inflated latency.
However, \oware using \sysname's SuspicionSensor can detect the delay attack and reassigns the leader role to exclude the malicious replica.
Consequently, at $t=137\,$s, \oware reconfigures to a new optimal configuration, thus restoring low latency, while BFT-SMaRt and Aware remain degraded.

\begin{figure}
	\centering
	\input{plots/opti-aware/calcTimes.tex}
	\caption{Suspicion graph computation time for \textsc{\oware}.}
	\label{fig:optaware:exp2}
\end{figure}

\subsection{\oware Scalability}
\label{eval:oware:scalability}
While the overhead of logging suspicions can be kept low using an efficient encoding scheme (see \cref{sect:overhead}), our primary concern here is evaluating the scalability of \oware's suspicion handling.
To this end, we measure the overhead of computing the maximum independent set in the resulting suspicion graph.
This step is necessary to identify candidate replicas for important consensus roles to be assigned during a reconfiguration.

\paragraph{Setup}
We generate random suspicion graphs for small ($n=4$) to medium-sized ($n=100$) system configurations relevant for PBFT.
Each graph encodes pairwise suspicions between replicas with edges representing mutual distrust.
For each value of $n$, we generate 100 random graphs to account for structural variance.
We compute the maximum independent set of each suspicion graph using a heuristic variant of the Bron-Kerbosch algorithm~\cite{bronkerbosch}, which detects cliques on the inverted graph.

\paragraph{Observations}
\cref{fig:optaware:exp2} plots the average time to compute the candidate set for increasing configuration sizes on a VM with two Intel Xeon Gold 6210U cores.
Execution times below 1~ms for smaller systems ($n<25$) are shown in the plot's inset.
While the computation time grows rapidly with $n$, it remains below 1~s for all configurations up to $n=100$, which we consider acceptable for typical PBFT-scale deployments.
This suggests the candidate selection process described in \cref{sec:mms} can be efficiently used for reconfigurations.

%% file: plots/opti-aware/exp1.tex
\definecolor{airforceblue}{rgb}{0.36, 0.54, 0.66}

\definecolor{verylightgray}{gray}{0.95}
\pgfplotsset{
	compat=1.11,
	legend image code/.code={
		\draw[very thick, mark repeat=2,mark phase=2]
		plot coordinates {
			(0cm,0cm)
			(0.15cm,0cm)        %
			(0.3cm,0cm)         %
		};%
	}
}

 \begin{tikzpicture}
    \begin{axis}[
width= 0.95\columnwidth,
height=3.5cm,
     tick label style={font=\footnotesize},
    label style={font=\footnotesize},
    legend style={font=\footnotesize},
    ylabel={Latency [ms]},
    xmin=0,
    xmax=200,
    ymin=0,
    ymax=1200,
    xtick={0, 30, 60, 90, 120, 150, 180},
    ytick={0, 200, 400, 600, 800,1000, 1200},
    ymajorgrids=true,
    yminorgrids=true,
    xminorgrids=true,
    minor grid style={dashed,gray!50},
    minor tick num=1,
    minor tick length=0.5ex,
    legend pos=south east,
    legend columns = 2,
    legend style={at={(0.99, 0.73)}},
    legend cell align={left},
    ymajorgrids=true,
    xmajorgrids=true,
    grid style=dashed,
    tick align = inside,
    legend style={font=\footnotesize},
]

\addplot[
    color=airforceblue,
    mark=.,
    ]
    table [x=time, y=latencyOptiLog] {plots/opti-aware/data/exp1.txt};
   \legend{\tiny \textsc{\oware}}
   \filldraw[black] (axis cs:82,50) circle[radius=1.5pt];
   \node[anchor=west, font=\tiny, text=black] at (axis cs:82,50) {delay attack};

   \draw[<->, thick, black] (axis cs:84,910) -- (axis cs:135,910)
   node[midway, below, font=\tiny, text=black] {log suspicions};

   \filldraw[black] (axis cs:137,240) circle[radius=1.5pt];
   \node[anchor=west, font=\tiny, text=black] at (axis cs:137,240) {reconfigure};
\end{axis}

\begin{axis}[
	at={(0.07\columnwidth,1.85cm)}, %
	anchor=north west,
	width=0.22\columnwidth,
	height=1cm,
	scale only axis,
	xmin=20, xmax=90,
	ymin=0, ymax=250,
	xtick={20, 40, 60, 80},
	    ymajorgrids=true,
	yminorgrids=true,
	    yminorgrids=true,
	        minor grid style={dashed,gray!50},
	    minor tick num=1,
	    minor tick length=0.5ex,
	xminorgrids=true,
	ytick={0,100, 200},
	tick label style={font=\tiny},
	label style={font=\tiny},
	title style={font=\tiny},
	axis background/.style={fill=verylightgray},
	clip=true
	]

	\addplot[
	color=airforceblue,
	mark=.,
	]
	table [x=time, y=latencyOptiLog] {plots/opti-aware/data/exp1.txt};
	   \filldraw[black] (axis cs:41,220) circle[radius=1pt];
	\node[anchor=west, font=\tiny, text=black] at (axis cs:42,220) {optimization};
\end{axis}
\end{tikzpicture}

\vskip -0.3cm

%% file: plots/opti-aware/exp1aware.tex
\definecolor{awesome}{rgb}{1.0, 0.13, 0.32}
\definecolor{cadet}{rgb}{0.03, 0.27, 0.49}

\definecolor{verylightgray}{gray}{0.95}
\pgfplotsset{
	compat=1.11,
	legend image code/.code={
		\draw[very thick, mark repeat=2,mark phase=2]
		plot coordinates {
			(0cm,0cm)
			(0.15cm,0cm)        %
			(0.3cm,0cm)         %
		};%
	}
}

 \begin{tikzpicture}
    \begin{axis}[
width= 0.95\columnwidth,
height=3.5cm,
     tick label style={font=\footnotesize},
    label style={font=\footnotesize},
    legend style={font=\footnotesize},
    xlabel={Time [s]},
    ylabel={Latency [ms]},
    xmin=0,
    xmax=200,
    ymin=0,
    ymax=1200,
    xtick={0, 30, 60, 90, 120, 150, 180},
    ytick={0, 200, 400, 600, 800,1000, 1200},
    ymajorgrids=true,
    yminorgrids=true,
    xminorgrids=true,
    minor grid style={dashed,gray!50},
    minor tick num=1,
    minor tick length=0.5ex,
    legend pos=south east,
    legend columns = 1,
    legend style={at={(0.99, 0.08)}},
    legend cell align={left},
    ymajorgrids=true,
    xmajorgrids=true,
    grid style=dashed,
    tick align = inside,
    legend style={font=\footnotesize},
]

\addplot[
    color=awesome,
    mark=.,
    ]
    table [x=time, y=latencyAWARE] {plots/opti-aware/data/exp1aware.txt};

   \addplot[
   color=cadet,
   mark=.,
   ]
   table [x=time, y=latencyBftsmart] {plots/opti-aware/data/exp1bftsmart.txt};

   \legend{\tiny \textsc{Aware}, \tiny \textsc{BFT-SMaRt/Pbft}}
\filldraw[black] (axis cs:82,50) circle[radius=1.5pt];
\node[anchor=west, font=\tiny] at (axis cs:82,50) {delay attack};

\end{axis}
\begin{axis}[
	at={(0.07\columnwidth,1.85cm)}, %
anchor=north west,
	width=0.22\columnwidth,
	height=1cm,
	scale only axis,
	xmin=20, xmax=90,
	ymin=0, ymax=250,
	xtick={20, 40, 60, 80},
	    ymajorgrids=true,
	yminorgrids=true,
	    yminorgrids=true,
	        minor grid style={dashed,gray!50},
	    minor tick num=1,
	    minor tick length=0.5ex,
	xminorgrids=true,
	ytick={0,100, 200},
	tick label style={font=\tiny},
	label style={font=\tiny},
	legend columns = 1,
	title style={font=\tiny},
	axis background/.style={fill=verylightgray},
	clip=true
	]

	\addplot[
	color=awesome,
	mark=.,
	]
	table [x=time, y=latencyOptiLog] {plots/opti-aware/data/exp1.txt};

	   \addplot[
	color=cadet,
	mark=.,
	]
	table [x=time, y=latencyBftsmart] {plots/opti-aware/data/exp1bftsmart.txt};

\end{axis}
\end{tikzpicture}

%% file: plots/opti-aware/calcTimes.tex
 \definecolor{verylightgray}{gray}{0.95}
\pgfplotsset{
	compat=1.11,
	legend image code/.code={
		\draw[very thick, mark repeat=2,mark phase=2]
		plot coordinates {
			(0cm,0cm)
			(0.15cm,0cm)        %
			(0.3cm,0cm)         %
		};%
	}
}

 \begin{tikzpicture}
    \begin{axis}[
width= 0.95\columnwidth,
height=3.5cm,
     tick label style={font=\footnotesize},
    label style={font=\footnotesize},
    legend style={font=\footnotesize},
    xlabel={Configuration size $n$},
    ylabel={Time \textbf{[ms]}},
    ymode=log,
    xmin=0,
    xmax=100,
    ymin=1,
    ymax=1000,
    xtick={0, 25, 50, 75, 100},
    ytick={1, 10, 100, 1000},
    ymajorgrids=true,
    yminorgrids=true,
    xminorgrids=true,
    yminorgrids=true,
	minor y tick num=1,
    minor grid style={dashed,gray!50},
    minor tick num=4,
    minor tick length=0.5ex,
    legend pos=south east,
    legend columns = 2,
    legend style={at={(0.99, 0.73)}},
    legend cell align={left},
    ymajorgrids=true,
    xmajorgrids=true,
    grid style=dashed,
    tick align = inside,
    legend style={font=\footnotesize},
    scaled y ticks=false,
tick scale binop=\times,
tick label style={/pgf/number format/fixed},
]

\addplot[    
    color=black,
    mark=o,
    mark options={fill=black},
	mark size=1pt, 
    error bars/.cd,
        y dir=both,         %
        y explicit, 
         error bar style={red},        %
    ]
table [
    x=n,
    y expr=\thisrow{time}/1000,
    y error expr=\thisrow{std}/1000
] {plots/opti-aware/data/calcTimes_us.txt};

    \end{axis}

\begin{axis}[
	at={(0.02\columnwidth,1.8cm)}, %
	anchor=north west,
	width=3.5cm,
	height=2.5cm,
	   ylabel={Time \textbf{[$\mathbf{\mu}$s]}},
	    ylabel near ticks, yticklabel pos=right,
    ymode=log,
    xmin=4,
    xmax=25,
    ymin=10,
    ymax=1000,
    xtick={4, 10, 16, 22},
    ytick={10, 100, 1000},
    ymajorgrids=true,
    yminorgrids=true,
    xminorgrids=true,
    yminorgrids=true,
	    minor grid style={dashed,gray!50},
	    minor tick num=1,
	    minor tick length=0.5ex,
	xminorgrids=true,
	tick label style={font=\tiny},
	label style={font=\tiny},
	title style={font=\tiny},
	axis background/.style={fill=verylightgray},
	clip=true
	]

\addplot[    
    color=black,
    mark=o,
    mark options={fill=black},
	mark size=1pt, 
    error bars/.cd,
        y dir=both,         %
        y explicit, 
         error bar style={red},        %
    ]
   table [
    x=n,
    y expr=\thisrow{time},
    y error expr=\thisrow{std}
] {plots/opti-aware/data/calcTimes_us.txt};

\end{axis}    
\end{tikzpicture} %

%% file: tex/evaluation_baseline.tex
\subsection{\optitree Baseline Comparison}
\label{eval:optitree:baseline}
This section evaluates \optitree compared to Kauri and HotStuff by measuring throughput and latency.
This evaluation shows the potential of applying \sysname to Kauri.
We conducted experiments with 21, 43, and 73 replicas, each with 4 clients.
Replicas were distributed across cities based on configuration size:
21 cities in Europe, 43 cities across Europe and North America, and 73 cities worldwide.
We used the same cities as in \cref{fig:optaware:exp1}.

For experiments with \optitree, we initially run simulated annealing for one second to search for a low-latency tree configuration.
For Kauri, tree configurations are randomly selected.
We configure the timers as described in \cref{sec:tree-suspicion} with $\delta=1$.
For comparison, Kauri's intermediate node timers are also set based on the recorded latencies \lm, rather than using a fixed timer~\cite{kauri}.
Each experiment runs for 120~seconds, with throughput and latency recorded every second.
We report the average throughput and latency over the 120~seconds, with error bars representing the 95\% confidence interval.

\begin{figure}[t]
  \begin{subfigure}{\columnwidth}
    \input{plots/fbasethroughput}
    \label{fig:basethroughput}
  \end{subfigure}
  \begin{subfigure}{\columnwidth}
    \input{plots/fbaselatency}

    \label{fig:baselatency}
  \end{subfigure}
  \caption{Throughput and latency of HotStuff, Kauri, and \optitree with different geographic distributions.}
  \Description{Plots showing throughput and latency of HotStuff, Kauri, and \optitree for different geographic distributions.}
  \label{fig:throughputlatency}
\end{figure}

\cref{fig:throughputlatency}  shows the throughput and latency results.
As expected, pipelined Kauri outperformed HotStuff, while \optitree achieved significantly better throughput than both protocols.
HotStuff-rr refers to HotStuff with round-robin leader selection, while HotStuff-fixed uses a random fixed leader.
The tree overlay in Kauri and \optitree trades reduced latency for higher throughput, but \optitree mitigates this to a large extend.
With 73 replicas distributed worldwide, \optitree increased throughput by 159\% and reduced latency by 39\%.
Given that practical deployments~\cite{solana, avalanche-network, cosmos, algorand} often involve hundreds of replicas distributed across diverse locations, \optitree provides a distinct advantage.

We repeated the experiment on a simulated Stellar~\cite{stellar} network.
Stellar provides a list of validators and their location.
At the time of the experiment, the Stellar network had 56 validators distributed worldwide.
We mapped these validator locations to the closest locations in our network emulator.
In this experiment, \optitree increased throughput by 67.5\% and reduced latency by 36\% compared to Kauri.

%% file: plots/fbasethroughput.tex
\centering
\begin{tikzpicture}
  \begin{axis}[
      width  = \columnwidth,
      height = 3.6cm,
      tick label style={font=\footnotesize},
      label style={font=\footnotesize},
      major x tick style = transparent,
      ymajorgrids=true,
      yminorgrids=true,
      legend cell align={left},
      minor grid style={dashed,gray!50},
      minor tick num=1,
      ybar=0.5pt,
      enlarge x limits=0.17,
      bar width=6pt,
      ymajorgrids = true,
      ylabel = {Throughput [Op/s]},
      symbolic x coords = {Europe21, NA-EU43, Stellar56, Global73},
      xtick = data,
      xticklabels = {Europe21, NA-EU43, Stellar56, Global73},
      minor tick length=0.5ex,
      tick align = inside,
      every node near coord/.append style={
        anchor=west,
        rotate=90
      },
      ymin=0,
      label style={font=\footnotesize},
      legend style={at={(1,1.1)},anchor=north east,font=\scriptsize},
      legend columns=2,
      legend entries = {\optitree, \optitree(no pipeline), Kauri (pipeline), HotStuff-rr, HotStuff-fixed}
    ]

    \addplot[style={draw=dark-blue ,fill=light-blue,mark=none, postaction={pattern=north east lines}}, error bars/.cd, y dir=both, y explicit, error bar style=red]
    coordinates {
      (Europe21, 33654.594566866355) +=(0, 487.74976898752357) -=(0,487.74976898752357)
      (NA-EU43,12330.309820547469) +=(0, 83.98476319981637) -=(0,83.98476319981637)
      (Stellar56, 12133.521940111104) +=(0, 134.95502843058784) -=(0,134.955028430587848)
    (Global73, 9084.235658542237) +=(0, 37.1809124885458) -=(0,37.1809124885458)};

    \addplot[style={draw=db3,fill=lb3,mark=none, postaction={pattern=dots}}, error bars/.cd, y dir=both, y explicit, error bar style=red]
    coordinates {
      (Europe21, 12318.676840425165) +=(0, 161.75216625177927) -=(0,161.75216625177927)
      (NA-EU43, 4198.421391673457)  +=(0, 31.61862583727725) -=(0,31.61862583727725)
      (Stellar56,4143.753772831834 ) +=(0, 53.472586409533) -=(0, 53.472586409533)
    (Global73, 2993.8699831717404)  +=(0,83.44534638635014) -=(0,83.44534638635014)};

    \addplot[style={draw=dark-green,fill=light-green,mark=none, postaction={pattern=horizontal lines}}, error bars/.cd, y dir=both, y explicit, error bar style=red]
    coordinates {
      (Europe21, 21888.662130836597) +=(0, 163.99346681202223) -=(0,163.99346681202223)
      (NA-EU43, 8952.992765711671) +=(0, 66.5381311985584) -=(0, 66.5381311985584)
      (Stellar56, 7244.739147391585)  +=(0, 30.136130265871543) -=(0, 30.136130265871543)
    (Global73, 3509.289268125633)  +=(0, 34.18760334167291) -=(0, 34.18760334167291)};

    \addplot[style={draw=dark-yellow,fill=light-yellow,mark=none}, error bars/.cd, y dir=both, y explicit, error bar style=red]
    coordinates {
      (Europe21, 17375.519643530373) +=(0, 163.99346681202223) -=(0,163.99346681202223)
      (NA-EU43, 7454.008458935862) +=(0, 66.5381311985584) -=(0, 66.5381311985584)
      (Stellar56, 3933.70449727075675)  +=(0, 15.569076278815828) -=(0, 15.569076278815828)
    (Global73, 4921.482831804681) +=(0, 47.38100068871154) -=(0, 47.38100068871154)};

    \addplot[style={draw=dark-red,fill=light-red,mark=none}, error bars/.cd, y dir=both, y explicit, error bar style=red]
    coordinates {
      (Europe21, 12122.015038325724) +=(0, 96.9909100328914) -=(0,96.9909100328914)
      (NA-EU43, 4240.697656835752) +=(0, 31.5381311985584) -=(0, 31.5381311985584)
      (Stellar56, 4171.802440580422) +=(0,49.17431716543979) -=(0, 49.17431716543979)
    (Global73, 3121.7043557058496)  +=(0, 18.794694672252717) -=(0, 18.794694672252717)};

  \end{axis}
\end{tikzpicture}

%% file: plots/fbaselatency.tex
\centering
\begin{tikzpicture}
  \begin{axis}[
      width  = \columnwidth,
      height = 3.6cm,
      tick label style={font=\footnotesize},
      label style={font=\footnotesize},
      major x tick style = transparent,
      ymajorgrids=true,
      yminorgrids=true,
      minor grid style={dashed,gray!50},
      minor tick num=1,
      ybar=0.5pt,
      enlarge x limits=0.17,
      bar width=6pt,
      ymajorgrids = true,
      ylabel = {Latency [s]},
      symbolic x coords = {Europe21, NA-EU43, Stellar56, Global73},
      xtick = data,
      xticklabels = {Europe21, NA-EU43, Stellar56, Global73},
      xlabel = {Configuration size \n},
      minor tick length=0.5ex,
      tick align = inside,
      every node near coord/.append style={
        anchor=west,
        rotate=90
      },
      ymin=0
    ]

    \addplot[style={draw=dark-blue ,fill=light-blue,mark=none, postaction={pattern=north east lines}}, error bars/.cd, y dir=both, y explicit, error bar style=red]
    coordinates {
      (Europe21, 0.31833330354423234) +=(0, 0.022089898238304206) -=(0,0.022089898238304206)
      (NA-EU43, 0.795794080228345) +=(0, 0.0029233321005476265) -=(0, 0.0029233321005476265)
      (Stellar56, 1.050200032083834) +=(0, 0.035720407908001794) -=(0, 0.035720407908001794)
    (Global73, 1.1831857081590877)+=(0,0.003156163926488542) -=(0,0.003156163926488542)};

    \addplot[style={draw=db3,fill=lb3,mark=none, postaction={pattern=dots}}, error bars/.cd, y dir=both, y explicit, error bar style=red]
    coordinates {
      (Europe21, 0.3016107637305245) +=(0, 0.019889642369398997) -=(0, 0.019889642369398997)
      (NA-EU43, 0.793993748463873) +=(0, 0.0029258932232787416) -=(0, 0.0029258932232787416)
      (Stellar56, 1.1208232398595783 ) +=(0, 0.039815158969099596) -=(0, 0.039815158969099596)
    (Global73, 1.1232170583992303)  +=(0,0.003866266736097579) -=(0,0.003866266736097579)};

    \addplot[style={draw=dark-green,fill=light-green,mark=none, postaction={pattern=horizontal lines}}, error bars/.cd, y dir=both, y explicit, error bar style=red]
    coordinates {
      (Europe21, 0.49469647657681703) +=(0, 0.0030159857695174486) -=(0,0.0030159857695174486)
      (NA-EU43, 1.1177166625459105) +=(0, 0.0044560265967075186) -=(0, 0.0044560265967075186)
      (Stellar56, 1.64157649551455)  +=(0, 0.0022713831671439255) -=(0, 0.0022713831671439255)
    (Global73, 1.8689663060949675) +=(0, 0.004315929086346593) -=(0, 0.004315929086346593)};

    \addplot[style={draw=dark-yellow,fill=light-yellow,mark=none}, error bars/.cd, y dir=both, y explicit, error bar style=red]
    coordinates {
      (Europe21, 0.25359201490936845) +=(0, 0.015779575925414097) -=(0, 0.015779575925414097)
      (NA-EU43, 0.5273596506738389) +=(0, 0.010936271420536103) -=(0, 0.010936271420536103)
      (Stellar56, 0.8929171811293428)  +=(0, 0.013212819404136988) -=(0, 0.013212819404136988)
    (Global73, 0.7888694746105988) +=(0, 0.0060266781138687) -=(0, 0.0060266781138687)};

    \addplot[style={draw=dark-red,fill=light-red,mark=none}, error bars/.cd, y dir=both, y explicit, error bar style=red]
    coordinates {
      (Europe21, 0.2727835109346747) +=(0, 0.00137392671073866) -=(0,0.00137392671073866)
      (NA-EU43, 0.7950001423548719) +=(0, 0.0031490298242566883) -=(0, 0.0031490298242566883)
      (Stellar56, 0.8372057640954461) +=(0, 0.004316285352429894) -=(0, 0.004316285352429894)
    (Global73, 1.0935454366557567)  +=(0, 0.003770639732434411) -=(0, 0.003770639732434411)};

  \end{axis}
\end{tikzpicture}

%% file: tex/security_latency.tex
\subsection{Suspicions' Impact on \optitree's Latency}
\label{sec:security-latency}

In this section we analyze how suspicions and failures affect the latency of \optitree's configurations.
We also examine how faulty replicas may exploit the suspicion mechanism by raising false suspicions.
Faulty replicas can exploit the tree formation process by pre-computing the optimal tree based on the recorded latencies and then raise suspicions against the correct internal nodes.
This strategy increases tree latency by targeted elimination of a few correct replicas, forming a suboptimal tree.
The severity of such an attack depends on the latency distribution of the replicas.

We simulated the attack to assess its impact on \optitree.
Starting with the assumption that all replicas are non-faulty, we first compute a tree using simulated annealing.
We then simulate the attack by randomly selecting an internal node to raise suspicion against the root, removing both from the candidate set.
The next tree is computed with the remaining replicas as candidates and an increased $\un$.
This attack is repeated $\f$ times to analyze how \optitree's performance degrades with increasing faults in the configuration.
We compare these results with the random trees formed in Kauri, which can only reconfigure up to $\sqn$ times before reverting to a star topology.

\begin{figure}[t]
  \centering
  \input{plots/sim-reconfig-3f}
  \caption{Tree latency (score) with increasing number of reconfigurations. Replicas are randomly distributed across the world.}
  \Description{Plot showing the tree latency w.r.t. reconfigurations.}
  \label{fig:sa-reconfig}
\end{figure}

\cref{fig:sa-reconfig} shows the results for a configuration with 211 replicas, with each data point representing the average of 1000 simulation runs with a 95\% confidence interval.
\optitree reconfigures if more than $\un$ votes are missing, while Kauri reconfigures if $\f$ votes are missing.
We therefore report the score for collecting $\q +\un$ replicas in \optitree and $\q+\f$ replicas in Kauri.
The increased latency observed in \optitree is due to the successive rise in $\un$ and the reduced availability of candidates.
A separate investigation (see \cref{app-eval:optitree:score}) showed that at 32 suspicions the reduced number of candidates causes 65\% of the latency increase, while the rest is due to a larger factor $\un$ in the score function.
The stepwise pattern reflects the higher latency required to collect votes from additional subtrees.

We observe that Kauri can benefit from simulated annealing to improve tree formation, so we implemented a variant called \textit{Kauri-sa}.
However, this variant does not benefit from the increasing $\un$ and must account for up to $\f$ failures.
In Kauri-sa, all internal nodes are excluded from the candidate set after each reconfiguration and a new tree is formed with the remaining replicas as candidate internal nodes.
The increased latency towards the last reconfiguration in Kauri-sa is due to the reduced number of available candidates.
In contrast, even after 35 suspicions have been recorded, \optitree retains enough candidates to construct a tree with latency comparable to the initial tree formed by Kauri-sa.

To summarize, \optitree's increased resilience compared to Kauri-sa is due to \sysname's suspicion monitoring, which provides both the candidate set and the estimate $\un$.

%% file: plots/sim-reconfig-3f.tex
\centering
\begin{tikzpicture}

  \def\datafileA{plots/reconfigdata/kauri-3f-14.csv}
  \pgfplotstableread{\datafileA}\datatable
  \def\datafileB{plots/reconfigdata/kauri-sa-3f-14.csv}
  \pgfplotstableread{\datafileB}\datatable
  \def\datafileC{plots/reconfigdata/opti-reconfig-2f-u.csv}
  \pgfplotstableread{\datafileC}\datatable

  \begin{axis}[
      width  = \columnwidth,
      height = 3.6cm,
      tick label style={font=\footnotesize},
      label style={font=\footnotesize},
      legend style={at={(1,1)},anchor=north east,font=\scriptsize},
      legend columns=2,
      ymajorgrids,
      scaled y ticks=base 10:-5,
      ytick scale label code/.code={},
      xlabel={Reconfigurations},
      ylabel={Latency [s]},
      legend entries={Kauri, Kauri-sa, \optitree},
      legend cell align={left},
      major x tick style = transparent,
      ymajorgrids=true,
      yminorgrids=true,
      minor tick length=0.5ex,
      tick align = inside,
      minor grid style={dashed,gray!50},
      minor tick num=1,
      ymin=120000,
      cycle list name=exotic,
    ]

    \addplot+[error bars/.cd, y fixed, y dir=both, y explicit] table[x=reconfig, y=latency, y error=error, col sep=comma]{\datafileA};
    \addplot+[error bars/.cd, y fixed, y dir=both, y explicit] table[x=reconfig, y=latency, y error=error, col sep=comma]{\datafileB};
    \addplot+[error bars/.cd, y fixed, y dir=both, y explicit] table[x=reconfig, y=latency, y error=error, col sep=comma]{\datafileC};

  \end{axis}
\end{tikzpicture}

%% file: tex/evaluation_tolerance.tex
\begin{figure}[tbp]
  \begin{subfigure}{\columnwidth}
    \input{plots/wriggle_th}
    \label{fig:wriggle-th}
  \end{subfigure}
  \begin{subfigure}{\columnwidth}
    \input{plots/wriggle_lt}
    \label{fig:wriggle-lt}
  \end{subfigure}
  \caption{Throughput and latency of \optitree in Europe21 with varying delays imposed by faulty internal nodes.}
  \Description{Plots showing throughput and latency of \optitree in Europe21 with varying delays imposed by faulty internal nodes.}
  \label{fig:wiggle-throughputlatency}
\end{figure}

\subsection{Tolerance to Malicious Delays}
\label{eval:optitree:tolerance}
As discussed in \cref{sec:suspicion}, \sysname uses the $\delta$ parameter to adjust timers based on recorded latencies.
While this mechanism is designed to tolerate latency variations, faulty replicas may artificially inflate their latency by a factor $\delta$ without raising suspicion.
We conducted experiments to assess the impact of such an attack on \optitree's performance.

In this experiment, we ran \optitree without pipelining.
The experiment involved 21 replicas forming a tree with branch factor 4.
We varied the number of faulty replicas from 1 to 4, randomly selecting them from the intermediate nodes in the tree.
These faulty replicas delay requests to leaf nodes and aggregation messages to the root by adjusting the $\delta$ parameter.
We tested $\delta$ values of 1.1, 1.2, and 1.4.

\cref{fig:wiggle-throughputlatency} shows the impact of faulty replicas on throughput and latency.
We see that a larger delay allows attackers to reduce performance significantly, reducing throughput by almost 49\%.
We note that, while a large $\delta$ enables this attack, a small $\delta$ can easily trigger false suspicions.

The above experiments use simulated network latencies.
These simulations do not include sudden latency spikes that may occur in real deployments.
While \sysname tolerates spikes within the $\delta$ parameter, our results highlight the trade-off in choosing $\delta$.
\cref{fig:sa-reconfig} illustrates how suspicions, that may be triggered by latency spikes, affect configuration latency.
\cref{fig:wiggle-throughputlatency} shows how faulty replicas can exploit a larger $\delta$ to slow down the system.
A smaller $\delta$ increases sensitivity to latency variations, causing more suspicions and reconfigurations, whereas a larger $\delta$ tolerates benign variations but allows faulty replicas to increase latency.

\sysname enables selecting an optimal $\delta$ through historical analysis of recorded latencies.
By systematically analyzing past latency data, \sysname could determine $\delta$ values best suited for varying network conditions, thereby enhancing system resilience and performance.
A comprehensive evaluation of this approach is left for future work.

%% file: plots/wriggle_th.tex
\centering
\begin{tikzpicture}
  \pgfplotsset{compat=newest}
  \def\datafileA{plots/wriggle/11-wriggle-th.csv}
  \pgfplotstableread{\datafileA}\datatable
  \def\datafileB{plots/wriggle/12-wriggle-th.csv}
  \pgfplotstableread{\datafileB}\datatable
  \def\datafileC{plots/wriggle/14-wriggle-th.csv}
  \pgfplotstableread{\datafileC}\datatable

  \begin{axis}[
      width  = \columnwidth,
      height = 3.5cm,
      tick label style={font=\footnotesize},
      label style={font=\footnotesize},
      legend style={font=\footnotesize},
      major x tick style = transparent,
      ymajorgrids=true,
      yminorgrids=true,
      minor grid style={dashed,gray!50},
      minor tick num=1,
      minor tick length=0.5ex,
      ybar=0.5pt,
      enlarge x limits=0.17,
      bar width=6pt,
      ylabel={Throughput [Op/s]},
      xtick={1,2,3,4},
      tick align = inside,
      every node near coord/.append style={
        anchor=west,
        rotate=90
      },
      ymin=0
    ]

    \addplot[style={draw=dark-green, fill=light-green, mark=none, postaction={pattern=dots}}, error bars/.cd, y dir=both, y explicit, error bar style=red]
    table [x=faults, y=throughput, y error plus=pluserr, y error minus=minuserr, col sep=comma] {\datafileA};

    \addplot[style={draw=dark-yellow, fill=light-yellow, mark=none}, error bars/.cd, y dir=both, y explicit, error bar style=red]
    table [x=faults, y=throughput, y error plus=pluserr, y error minus=minuserr, col sep=comma] {\datafileB};

    \addplot[style={draw=dark-red, fill=light-red, mark=none}, error bars/.cd, y dir=both, y explicit, error bar style=red]
    table [x=faults, y=throughput, y error plus=pluserr, y error minus=minuserr, col sep=comma] {\datafileC};

    \addplot[red,thick,sharp plot,update limits=false,] coordinates { (0,12318.67) (5,12318.67) };

  \end{axis}
\end{tikzpicture}

%% file: plots/wriggle_lt.tex
\centering
\begin{tikzpicture}
  \pgfplotsset{compat=newest}
  \def\datafileA{plots/wriggle/11-wriggle-lt.csv}
  \pgfplotstableread{\datafileA}\datatable
  \def\datafileB{plots/wriggle/12-wriggle-lt.csv}
  \pgfplotstableread{\datafileB}\datatable
  \def\datafileC{plots/wriggle/14-wriggle-lt.csv}
  \pgfplotstableread{\datafileC}\datatable

  \begin{axis}[
      width  = \columnwidth,
      height = 3.5cm,
      tick label style={font=\footnotesize},
      label style={font=\footnotesize},
      major x tick style = transparent,
      ymajorgrids=true,
      yminorgrids=true,
      minor grid style={dashed,gray!50},
      minor tick num=1,
      minor tick length=0.5ex,
      ybar=0.5pt,
      enlarge x limits=0.17,
      bar width=6pt,
      ylabel={Latency [s]},
      xtick={1,2,3,4},
      xlabel={Faulty Internal Replicas},
      tick align = inside,
      every node near coord/.append style={
        anchor=west,
        rotate=90
      },
      ymin=0,
      legend style={at={(0,1)},anchor=north west,font=\scriptsize},
      legend columns=4
    ]

    \addlegendentry{1.1}
    \addlegendimage{ybar, fill=light-green, draw=dark-green}
    \addlegendentry{1.2}
    \addlegendimage{ybar, fill=light-yellow, draw=dark-yellow}
    \addlegendentry{1.4}
    \addlegendimage{ybar, fill=light-red, draw=dark-red}
    \addlegendentry{No Faults}
    \addlegendimage{line legend, thick, red, sharp plot}

    \addplot[red,thick,sharp plot,update limits=false] coordinates { (0,0.2616107) (5,0.2616107) };

    \addplot[style={draw=dark-green, fill=light-green, mark=none, postaction={pattern=dots}}, error bars/.cd, y dir=both, y explicit, error bar style=red]
    table [x=faults, y=latency, y error plus=pluserr, y error minus=minuserr, col sep=comma] {\datafileA};

    \addplot[style={draw=dark-yellow, fill=light-yellow, mark=none}, error bars/.cd, y dir=both, y explicit, error bar style=red]
    table [x=faults, y=latency, y error plus=pluserr, y error minus=minuserr, col sep=comma] {\datafileB};

    \addplot[style={draw=dark-red, fill=light-red, mark=none}, error bars/.cd, y dir=both, y explicit, error bar style=red]
    table [x=faults, y=latency, y error plus=pluserr, y error minus=minuserr, col sep=comma] {\datafileC};

  \end{axis}
\end{tikzpicture}

%% file: tex/evaluation_sa.tex
\begin{figure}[tbp]
    \centering
    \input{plots/sa-perf-new}
    \caption{Tree latency improves with longer search time.}
    \Description{Plot showing the tree latency improvement with increasing simulation time.}
    \label{fig:sa-perf}
  \end{figure}

\subsection{Impact of Simulated Annealing on Tree Latency}
\label{eval:optitree:sa}
In this section, we analyze the impact of \sysname's use of simulated annealing (SA) on tree latency in \optitree.
SA stops when the configured \textit{search time} expires, or the cooling temperature reaches a threshold.
We conducted experiments to evaluate how different search times impact the resulting tree's latency.
We vary the tree size from 21 to 211 replicas, and the search time is varied from 250~ms to 4~seconds, doubling with each step.
Each experiment was run 1000 times, and we report the average tree latency with error bars showing the 95\% confidence interval.

\cref{fig:sa-perf} shows how SA improves latency with increasing search time for different configuration sizes.
For smaller tree sizes (21 to 57 replicas), SA found no improvements beyond one second of search time.
For larger tree sizes, the latency improves with increased search time.
For a tree with 211 replicas, a 4-second search resulted in a tree with 35\% lower latency than a 250~ms search.

Since SA is a non-deterministic optimization technique, the variance in latency provides insights into its consistency.
As the search time increases, this variance decreases, indicating that longer search times improve the likelihood of consistently finding better tree configurations.
In addition, we also evaluated SA's impact on throughput during reconfigurations, as detailed in \cref{eval:optitree:reconfig}.

%% file: plots/sa-perf-new.tex
\centering
\begin{tikzpicture}

  \def\datafileA{plots/sa-results/bf-6.csv}
  \pgfplotstableread{\datafileA}\datatable
  \def\datafileB{plots/sa-results/bf-7.csv}
  \pgfplotstableread{\datafileB}\datatable
  \def\datafileC{plots/sa-results/bf-8.csv}
  \pgfplotstableread{\datafileC}\datatable
  \def\datafileD{plots/sa-results/bf-9.csv}
  \pgfplotstableread{\datafileD}\datatable
  \def\datafileE{plots/sa-results/bf-10.csv}
  \pgfplotstableread{\datafileE}\datatable
  \def\datafileF{plots/sa-results/bf-11.csv}
  \pgfplotstableread{\datafileF}\datatable
  \def\datafileG{plots/sa-results/bf-12.csv}
  \pgfplotstableread{\datafileG}\datatable
  \def\datafileH{plots/sa-results/bf-13.csv}
  \pgfplotstableread{\datafileH}\datatable
  \def\datafileI{plots/sa-results/bf-14.csv}
  \pgfplotstableread{\datafileI}\datatable

  \begin{axis}[
      width  = \columnwidth,
      height = 3.6cm,
      tick label style={font=\footnotesize},
      label style={font=\footnotesize},
      legend style={at={(1,1.2)},anchor=north east,font=\scriptsize},
      legend columns=3,
      ymajorgrids,
      scaled y ticks=base 10:-5,
      ytick scale label code/.code={},
      xlabel={Simulated Annealing Search Time [s]},
      ylabel={Latency [s]},
      legend entries={$n=57$,$n=91$,$n=111$,$n=157$,$n=183$,$n=211$},
      legend cell align={left},
      major x tick style = transparent,
      ymajorgrids=true,
      yminorgrids=true,
      minor tick length=0.5ex,
      tick align = inside,
      minor grid style={dashed,gray!50},
      minor tick num=1,
      ymin=90000,
      cycle list name=exotic,
    ]

    \addplot+[error bars/.cd, y fixed, y dir=both, y explicit] table[x=bf, y=latency, y error=err, col sep=comma]{\datafileB};
    \addplot+[error bars/.cd, y fixed, y dir=both, y explicit] table[x=bf, y=latency, y error=err, col sep=comma]{\datafileD};
    \addplot+[error bars/.cd, y fixed, y dir=both, y explicit] table[x=bf, y=latency, y error=err, col sep=comma]{\datafileE};
    \addplot+[error bars/.cd, y fixed, y dir=both, y explicit] table[x=bf, y=latency, y error=err, col sep=comma]{\datafileG};
    \addplot+[error bars/.cd, y fixed, y dir=both, y explicit] table[x=bf, y=latency, y error=err, col sep=comma]{\datafileH};
    \addplot+[error bars/.cd, y fixed, y dir=both, y explicit] table[x=bf, y=latency, y error=err, col sep=comma]{\datafileI};

  \end{axis}
\end{tikzpicture}

%% file: tex/evaluation_overhead.tex
  \begin{figure}[tbp]
  \vspace{-0.1cm}
    \centering
    \input{plots/overhead_new}
    \caption{Proposal size including different measurements.}
    \Description{Plot showing the average proposal size with various sensors activated.}
    \label{fig:byteoverhead}
  \end{figure}

\subsection{\sysname Overhead}
\label{sect:overhead}
This section evaluates the increase in proposal size due to \sysname's sensors.
\cref{fig:byteoverhead} shows the average proposal size when measurements from different sensors are included, for 20, 40, 60, and 80 replicas distributed across 10 locations.
With 80 replicas, including both latency measurements and suspicions increases the proposal size by about 270 bytes compared to the baseline.
Adding proofs of misbehavior further increases the proposal size by approximately 4.5~KB.
The \missensor includes proofs of misbehavior such as quorum certificates or signatures.
Unlike latency measurements, reports from the \missensor and \sussensor are infrequent, since complaints are raised at most once per replica.
With all sensors active, throughput and latency measurements showed no significant differences in WAN settings, while in a data center environment throughput was reduced by less than 1\% and CPU usage increased by 20\%.

%% file: plots/overhead_new.tex
\centering
\begin{tikzpicture}
  \begin{axis}[
      width  = \columnwidth,
      height = 3.6cm,
      tick label style={font=\footnotesize},
      label style={font=\footnotesize},
      major x tick style = transparent,
      ymajorgrids=true,
      yminorgrids=true,
      legend cell align={left},
      minor grid style={dashed,gray!50},
      minor tick num=1,
      ybar=0.5pt,
      enlarge x limits=0.17,
      bar width=6pt,
      ymajorgrids = true,
      ylabel = {Proposal size [bytes]},
      symbolic x coords={20,40,60,80},
      xtick = data,
      xticklabels = {20,40,60,80},
      xlabel = {Configuration size \n},
      minor tick length=0.5ex,
      tick align = inside,
      every node near coord/.append style={
        anchor=west,
        rotate=90
      },
      ymin=0,
      legend style={at={(0.75,1.32)},anchor=north east,font=\scriptsize},
      legend columns=2,
      legend entries = {Misbehavior+lv, Suspicion+lv, Latency vector (lv), No \sysname}
    ]

    \addplot[style={draw=dark-green, fill=light-green, mark=none, postaction={pattern=north east lines}}]
    coordinates {(20, 10235.734317343173) (40, 12211.060446780552) (60, 14477.091568449683) (80, 16451.324475524474)};

    \addplot[style={draw=db3, fill=lb3, mark=none, postaction={pattern=dots}}]
    coordinates {(20, 9182.68691189050475) (40, 10122.28143133462) (60, 11267.7134646962242) (80, 12186.331005586591)};

    \addplot[style={draw=dark-yellow, fill=light-yellow, mark=none}]
    coordinates {(20, 9115.662790697674) (40,10099.602) (60, 11224.836467702371) (80, 12138.74720670391)};

    \addplot[style={draw=dark-orange, fill=light-orange, mark=none}]
    coordinates {(20, 9095.813284864507) (40, 10014.957170668396) (60, 11002.0206185567) (80, 11918.053941908714)};

  \end{axis}
\end{tikzpicture}

%% file: tex/related.tex
\section{Related Work}

Several prior works~\cite{p2p-acc, network-acc,bft-forensics} use accountability as a tool to enhance security by supporting forensics to identify misbehaving replicas.
PeerReview~\cite{peerreview} introduced a comprehensive method to establish accountability in distributed systems by auditing messages from other replicas.
However, due to the extensive message exchange among replicas, this approach may be too costly for RSMs.
Polygraph~\cite{polygraph} introduced accountable Byzantine agreement by logging protocol messages to justify the decisions made by replicas.
IA-CCF~\cite{ia-ccf} improves on Polygraph by providing individual accountability through the use of \textit{ledgers} and \textit{receipts}.
It relies on trusted execution environments to secure replicas.
Given a set of conflicting receipts and a ledger, any third party can generate proof-of-misbehavior against at least $\n/3$ replicas.
These systems are designed for forensic analysis and do not utilize detected misbehavior to select a suitable configuration.
\sysname instead focuses on utilizing detected misbehavior or suspicions for informed configuration decisions.
While proof-of-misbehavior is a useful part of \sysname, we show that raising suspicions is sufficient to find performant configurations.

Several systems have optimized configuration latency by selecting leaders based on location~\cite{archer, droopy, mencius} or by ranking configurations to identify the optimal one~\cite{wheat, aware, eval}.
However, these approaches lack mechanisms to suspect or hold misbehaving replicas accountable for failed or degraded configurations.
Systems such as Archer~\cite{archer} and Aardvark~\cite{clement2009making} hold only the leader responsible.
In contrast, \sysname's \sussensor detects underperforming replicas, not only leaders, enabling its use in complex configurations such as tree topologies~\cite{kauri}, which may fail despite a correct leader.
RBFT~\citep{rbft} monitors performance through redundant leader executions, allowing it to detect slow leaders and replace misbehaving ones that deliberately degrade performance.
\sysname achieves similar detection without redundant executions, making it more practical for tree-based BFT protocols where the configuration space is vast.

Abstract~\cite{Abstract} introduced a switching mechanism to enable more efficient protocols under favorable conditions and to revert to more resilient protocols when conditions deteriorate.
Similarly, Chenyuan et al.~\cite{adaptive-bft} propose a system to collect and replicate configuration metrics and use them in a reinforcement learning-based mechanism to switch between protocols.
Adapt~\cite{adapt} presents a quality-control system that assesses the suitability of a configuration and protocol based on the current environment.
These systems treat the underlying protocols as black boxes.
Protocol switching is a useful fallback mechanism that complements intra-protocol optimizations, such as the role-assignment optimizations enabled by \sysname.

There are several techniques~\cite{king, idmaps, gnp} for estimating latency among Internet-deployed nodes.
Some approaches~\cite{idmaps} use landmark nodes to estimate latency, while others~\cite{king} rely on DNS name servers.
GNP~\cite{gnp} estimates unknown path latencies by embedding nodes in a geometric space and using Euclidean distance as the predictor.
Variants of these approaches incorporate machine learning, quality-of-service protocols, and application-specific mechanisms.
However, latency estimation is challenging in the presence of faulty nodes and attacks~\cite{robustvirtualcoordsys}. 
Therefore, \sysname employs direct latency measurements on all links.

%% file: tex/conclusion.tex
\section{Conclusion}
RSM optimizations are difficult to maintain in dynamically changing environments, especially in the presence of faulty replicas and potentially manipulated measurements.

\sysname's sensors and monitors tackle this problem by combining collection of local measurement and the execution of non-deterministic computations, with consistent presentations and processing of this information.
Moreover, \sysname selects candidates for special roles and holds faulty replicas accountable to the reported measurements.
For instance, by using \sysname's sensors and monitors, \oware 
can detect-and-mitigate performance attacks and maintain optimal latency. 
Furthermore, we showcase how \optitree leverages \sysname to find more performant configurations and to stay optimized despite failures. 
In a global deployment, \optitree can select tree configurations with up to 39\% lower latency and achieve $2.5\times$ of Kauri's throughput. 
In particular, the estimate on actual faults, provided by \sysname, enables \optitree to balance performance and robustness. %

%% file: tex/acknowledgements.tex
\begin{acks}
    We would like to thank our shepherd, Pierre-Louis Aublin, and the anonymous reviewers of OSDI 2024, EuroSys 2025 and 2026, for their insightful comments and suggestions that helped improve our paper.
    We also want to thank Arian Baloochestani for participating in early discussions and providing valuable feedback on the initial drafts of the paper.
    This work is partially funded by the BBChain and Credence projects under grants 274451 and 288126 from the Research Council of Norway.
    This work is partially funded by Deutsche Forschungsgemeinschaft (DFG, German Research Foundation) -- \href{https://gepris.dfg.de/gepris/projekt/446811880}{446811880 (BFT2Chain)}.
\end{acks}

%% file: tex/notation.tex
\section {Notation}
\begin{table}[ht]
  \footnotesize
  \centering
  \begin{tabular}{@{}ll@{}}
    \toprule
    \textbf{Notation}              & \textbf{Description} \\
    \midrule
    \n & total number of replicas in the configuration\\
    \q & number of replicas in a quorum\\
    \f & number of failures tolerated \\
    \tn & actual faults in the configuration \\
    \un & estimate on number of non-crash faults in the configuration\\
    \bn & branching factor of the tree \\
    \inn & the number of internal nodes \\
    \winlen & window length for sliding window mechanism\\
    \kn & number of nodes from which the tree's root collects votes\\
    \na,\nb,\nc,\nd & replicas \\
    \nr & root node of a tree \\
    \nin & intermediate node of a tree \\
    \nleaf{} & leaf node of a tree \\
    \nl & leader node \\
    \naT & a node in a triangle \\
    \nbC & a crashed node\\
    \sens{1} & sensor type 1 \\
    \mons{1} & monitor type 1 \\
    \C{} & the set of all $n$ replicas \\
    \Cfg & a configuration proposal \\
    \Crash & set of crashed replicas \\
    \Faulty & set of provably faulty replicas \\
    \Cand & set of replicas that are candidates for special roles \\
    \G & suspicion graph\\
    \V & set of vertices in \G\\
    \E & set of edges in \G\\
    \Q & quorum of nodes in a configuration\\
    \tree & a tree\\
    \Intermediate & set of all intermediate nodes in a tree\\
    \Internal & set of all internal nodes in a tree\\
    \Leaf & set of all leaf nodes in a tree\\
    \IS & independent set of nodes\\
    \MG & maximal set of disjoint nodes with two-way suspicions\\
    \T & set of nodes that form a triangle with two nodes in $\MG$  \\
    \susp{\na}{\nb} & replica $\na$ suspects replica $\nb$\\
    $\rdur$ & expected round duration\\
    $\dm$ & expected delay from round start to message arrival\\
    $\lat[r]{\na}{\nb}$ & recorded latency between replicas $\na$ and $\nb$\\
    $\lat[a]{\na}{\nb}$ & actual latency between replicas $\na$ and $\nb$\\
    $\agt{\nin}$ & time to aggregate votes at intermediate node $\nin$\\
    $\scorefn$ & general score function\\
    $\score(\un,\tree)$ & the minimum latency to collect votes from $\un+q$ nodes\\
    $\lm$ & latency matrix \\
    $\lme{\na}{\nb}$ & latency matrix element between replicas $\na$ and $\nb$\\
    $\delta$ & multiplier used in timer calculations\\
    \bottomrule
  \end{tabular}
\end{table}

%% file: tex/evaluation_score_app.tex
\section{Additional Results}
\label{sec:additional-results}
\subsection{Tree Robustness and Configuration Latency}
\label{app-eval:optitree:score}

When searching for a low-latency tree, one has to choose between optimizing the latency for fault-free cases, or optimizing the latency in presence of faults, resulting in a lower bound for fault-free cases.
\optitree uses the parameter \un to create increasingly robust trees, optimizing for the latency when \un leafs become unresponsive.
Recall that \un is the estimated number of non-crash faults provided by the \susmonitor.
We achieve this by using $\score(\kn,\tree)$ (see \cref{sec:ol-trees}) in simulated annealing.

\begin{figure}[t]
  \centering
  \input{plots/score_perf}
  \caption{Tree latency (score) degradation as the number of faulty nodes increase. Replicas are randomly distributed across the world.}
  \Description{Plot showing the tree latency degradation with more faulty nodes.}
  \label{fig:scoring-perf}
\end{figure}

We evaluated the cost of this overprovisioning, by varying \un from 5\% to 30\% of the tree size.
\cref{fig:scoring-perf} shows the best score when applying simulated annealing with increasing \un in the score function.
We use latency data for 220 cities from WonderProxy~\cite{wonderproxy}.
With a tree size of 211, latency increases by 54\% when \un is 30\% of the tree size.
We note that \optitree overprovisions only in the presence of adversarial omission and timing faults, not for pure crash faults.

%% file: plots/score_perf.tex
\centering
\begin{tikzpicture}

    \def\datafileA{plots/score_results/4-score.csv}
    \pgfplotstableread{\datafileA}\datatable
    \def\datafileB{plots/score_results/5-score.csv}
    \pgfplotstableread{\datafileB}\datatable
    \def\datafileD{plots/score_results/7-score.csv}
    \pgfplotstableread{\datafileD}\datatable
    \def\datafileE{plots/score_results/8-score.csv}
    \pgfplotstableread{\datafileE}\datatable
    \def\datafileG{plots/score_results/10-score.csv}
    \pgfplotstableread{\datafileG}\datatable
    \def\datafileH{plots/score_results/11-score.csv}
    \pgfplotstableread{\datafileH}\datatable
    \def\datafileI{plots/score_results/12-score.csv}
    \pgfplotstableread{\datafileI}\datatable
    \def\datafileJ{plots/score_results/13-score.csv}
    \pgfplotstableread{\datafileJ}\datatable
    \def\datafileK{plots/score_results/14-score.csv}
    \pgfplotstableread{\datafileK}\datatable

    \begin{axis}[
        width  = \columnwidth,
        height = 3.6cm,
        tick label style={font=\footnotesize},
        label style={font=\footnotesize},
        legend style={at={(0.48,1.39)},anchor=north,font=\scriptsize},
        legend columns=3,
        ymajorgrids,
        xlabel={Percentage of faulty leafs tolerated (\un)},
        ylabel={Latency [s]},
        legend entries={$n=21$,$n=43$,$n=91$,$n=111$,$n=157$,$n=211$},
        legend cell align={left},
        major x tick style = transparent,
        ytick={100000,150000,200000,250000},
        yticklabels={1,1.5,2,2.5},
        ymajorgrids=true,
        yminorgrids=true,
        minor tick length=0.5ex,
        tick align = inside,
        minor grid style={dashed,gray!50},
        minor tick num=1,
        ymin=90000,
        cycle list name=exotic,
        scaled y ticks=false,
    ]

    \addplot+[error bars/.cd, y fixed, y dir=both, y explicit] table[x=faults, y=latency, y error=err, col sep=comma]{\datafileA};
    \addplot+[error bars/.cd, y fixed, y dir=both, y explicit] table[x=faults, y=latency, y error=err, col sep=comma]{\datafileE};
    \addplot+[error bars/.cd, y fixed, y dir=both, y explicit] table[x=faults, y=latency, y error=err, col sep=comma]{\datafileG};
    \addplot+[error bars/.cd, y fixed, y dir=both, y explicit] table[x=faults, y=latency, y error=err, col sep=comma]{\datafileH};
    \addplot+[error bars/.cd, y fixed, y dir=both, y explicit] table[x=faults, y=latency, y error=err, col sep=comma]{\datafileI};
    \addplot+[error bars/.cd, y fixed, y dir=both, y explicit] table[x=faults, y=latency, y error=err, col sep=comma]{\datafileK};

    \end{axis}
\end{tikzpicture}

%% file: tex/evaluation_reconfig_app.tex
\subsection{\optitree Reconfiguration}
\label{eval:optitree:reconfig}

\begin{figure}[b]
    \centering
    \input{plots/faults-sa}
    \caption{Reconfiguration experiment with 21 replicas. The root replica fails every 10 seconds, triggering a simulated annealing search.}
    \Description{Plots showing results of a reconfiguration experiment with 21 replicas forming a tree with branching factor 4. The root replica fails every 10 seconds, triggering a simulated annealing search.}
    \label{fig:fault-1-sa}
\end{figure}

We now evaluate \optitree's reconfiguration mechanism.
We consider a scenario with 21 Europe-based replicas, where the root replica becomes faulty and triggers a reconfiguration every 10 seconds.
To make reconfigurations more visible, we do not process any payload, when recording suspicions and searching for a new tree.
\cref{fig:fault-1-sa} shows the throughput observed by a non-faulty replica.
We see that after one second used for simulated annealing, throughput is recovered.
This experiment shows one of the main tradeoffs in \optitree.
\optitree achieves better configurations but incurs higher reconfiguration latency than Kauri, when recovering from failed configurations.
In scenarios where all replicas are equidistant with each other, alternatives like Kauri may perform comparably or better.

%% file: plots/faults-sa.tex
\centering
\begin{tikzpicture}
    \pgfplotsset{compat=newest}
    \def\datafileA{plots/faultsdata/fault-1-th.csv}
    \pgfplotstableread{\datafileA}\datatable
    \begin{axis}[
                    height=3.6cm,
                    width=\linewidth,
        tick label style={font=\footnotesize},
       label style={font=\footnotesize},
       legend style={font=\footnotesize},
       major x tick style = transparent,
      ymajorgrids=true,
      yminorgrids=true,
        minor tick length=0.5ex,
       tick align = inside,
      minor grid style={dashed,gray!50},
      minor tick num=1,
                    xmin=0,xmax=90,
                    ymin=0,
                    xtick={0,10,20,...,90},
                    xlabel={Time [s]},
                    ylabel={Throughput [Op/s]},
            ]
    \addplot+[color=black, mark=+, mark size=1.5pt] table[x=time, y=commands, col sep=comma]{\datafileA};
    \end{axis}
    \end{tikzpicture}
        \vspace{-0.2cm}

%% file: tex/suspicion_proof_app.tex
\section{\sysname Correctness}
\label{sec:suspicion_proof}

In this section, we prove correctness of the suspicions monitoring in \sysname, as described in \cref{sec:suspicion}.
We proof the following guarantees:
\begin{enumerate}[label=C\arabic*]
  \item\label{p:candidates} There are always at least $\n-\f$ candidates available in $\Cand$. (\cref{lem:candidates})
  \item\label{p:suspicions} If some correct replica does not commit a new proposal on average, every $\df{\rdur}$, a suspicion is raised. (\cref{lem:suspicions})
  \item\label{p:no-susp} After GST, any suspicion between two correct replicas is filtered in the \susmonitor. (\cref{lem:c2c})
\end{enumerate}

These guarantees ensure that eventually (after GST) the system performs at a steady pace.
This pace can be optimized by searching for a role assignment minimizing $\rdur$, within the candidate set, as is done by our \configsensor.

\paragraph{\ref{p:candidates} -- Sufficient Candidates}
\cref{lem:candidates} ensures \ref{p:candidates}.
This is necessary to ensure that the system can always form a new configuration.

\begin{lemma}
  \label{lem:candidates}
  There are always at least $\n-\f$ candidates available in $\Cand$.
\end{lemma}
\begin{proof}
    As described in \cref{sec:suspicion}, the \susmonitor selects the candidate set $\Cand$ as a maximum independent set in the suspicion graph $\G$.
    Note that, if all suspicions are caused by faulty replicas, then the set of correct replicas is an independent set of size $\n-\f$.
    If no independent set of size $\n-\f$ can be found, old suspicions are removed from $\G$.
\end{proof}

\paragraph{\ref{p:suspicions} -- Suspicion are raised}
Claims \ref{p:suspicions} (\cref{lem:suspicions}) expresses that suspicions are raised, if the system does not make progress as expected.

\cref{tab:suspicion} repeats the conditions for the \sussensor to raise suspicions, as described in \cref{sec:suspicion}.
We repeat it here for better readability.

\begin{table}[ht]
    \caption{Suspicion conditions in \sysname}
    \label{tab:suspicion}
    \footnotesize
    \centering
    \begin{tabular}{@{}p{0.1cm}p{5.25cm}p{1.78cm}@{}}
      \toprule
      \multicolumn{2}{@{}l}{\textbf{Condition for Suspicion}}       & \textbf{Suspicion} \\ \midrule
      (a) & Proposal timestamps not within $\df{\rdur}$ & \suspm{Slow}{\na}{\nl} \\
      (b) & No message $m$ from $\nb$, $\df{\dm}$ after round start & \suspm{Slow}{\na}{\nb} \\
      (c) & False suspicion from $\nb$ on $\na$ & \suspm{False}{\na}{\nb} \\
      \bottomrule
    \end{tabular}
  \end{table}

\begin{lemma}
  \label{lem:suspicions}
  If some correct replica does not complete a round, every $\df{\rdur}$, a suspicion is raised.
\end{lemma}
\begin{proof}
  Let $\na$ be a correct replica, that commits a proposal in round $r_i$ after receiving messages $m_1, ... m_k$ and let $t_i^p$ be the proposal timestamp of round $r_i$.
  If any message $m\in \{m_1, ... m_k\}$ is not received by $\na$ at most $\df{\dm}$ after $t_i^p$, then the \sussensor raises a suspicion following condition (b).
  Additionally, the \sussensor raises a suspicion (b) if consecutive proposal timestamps are not within $\df{\rdur}$, i.e. if $t_{i+1}^p - t_i^p > \df{\rdur}$.
  Thus, if no suspicions are raised by $\na$ proposal timestamps appear at least every $\df{\rdur}$ and rounds are finished timely after the proposal timestamp.
\end{proof}

We note that Claim \ref{p:suspicions} follows from \cref{lem:suspicions} only for protocols that commit a proposal in one round.
This is the case for PBFT based protocols.
In HotStuff and derived protocols, like Kauri and \optitree a proposal is not committed in one round.
Instead in steady operation in every round, the proposal of a previous round is committed. Thus the claim holds also here.

\paragraph{\ref{p:no-susp} -- Filtered false suspicions}
\cref{lem:c2c} proves \ref{p:no-susp}, ensuring that suspicions which are not filtered after GST are useful.

Clearly, to avoid false suspicions between correct replicas requires that timeout durations $\dm$ and $\rdur$ are set correctly.
The timeout requirements listed as \ref{tr:dml}-\ref{tr:rdur} below give precise requirements for these timeouts.
Below we argue how these requirements are satisfied in \oware and \optitree.

\begin{enumerate}[label=TR\arabic*]
    \item\label{tr:dml} Let $m$ be a message sent from the leader $\nl$ to replica $\na$ directly after creating the proposal. Then $\dm$ is configured correctly if: $\dm = \lat[r]{\nl}{\na}$.
    \item\label{tr:dm} Let $m$ be a message sent from replica $\na$ to replica $\nb$.
    Then there exist an earlier message $m'$ sent to $\na$, such that $\dm=d_{m'} + \lat[r]{\na}{\nb}$.

    \item\label{tr:rdur} There exists a message $m$ sent to the leader $\nl$, such that $\rdur=d_m$.
\end{enumerate}

\begin{example} \textbf{Timeouts in \oware:}
In \oware, the timeout $\dm$ and $\rdur$ are set according to the timeout requirements. Below we discuss how to set timeouts for individual messages and the round timeout. We use message types from the Aware (BFT-Smart) protocol, but also mention equivalent types from PBFT.
\begin{itemize}
    \item For a Propose message (Preprepare in PBFT) $m_p$, $d_{m_p} = \lat[r]{\nl}{\na}$ and thus \ref{tr:dml} holds.
    \item Write messages in Aware (Prepare in PBFT) are sent after receiving a Propose message. For a Write message $m_w$ from $\na$ to $\nb$,
    $d_{m_w} = \lat[r]{\nl}{\na} + \lat[r]{\na}{\nb}$.
    Thus, \ref{tr:dm} holds for $m_w$ with $m'=m_p$.
    \item Accept messages in Aware (Commit in PBFT) are sent after receiving a quorum of Write messages. Thus, for an Accept message $m_a$ sent by $\na$ to $\nb$, the following equation holds, where \textit{quorums} is the set of all quorums of Write messages sent to $\na$:
    \[
      \dm = \min_{q \in \textit{quorums}} \left(\max_{m_w \in q}d_{m_w}\right) + \lat[r]{\na}{\nb}
    \]
    Thus, \ref{tr:dm} holds for $m_a$ with $m'$ being the write message with the slowest expected duration, in the fastest quorum of write message sent to $\na$.
    \item Finally, in Aware, the round is concluded when the leader $\nl$ receives a quorum of Accept messages. Thus the round duration $\rdur$ can be computed as follows:, where \textit{quorums} is the set of all quorums of Accept messages sent to $\nl$:
    \[
        \rdur = \min_{q \in \textit{quorums}} \left(\max_{m_a \in q}d_{m_a}\right)
    \]
    Thus, \ref{tr:rdur} holds with $m'$ being the accept message with the slowest expected duration, in the fastest quorum of accept messages sent to $\nl$.
\end{itemize}
We note that the $\rdur$ developed above is the same as the result of the score function defined by Aware.
These timeouts are used in \oware.
\end{example}

\begin{lemma}
  \label{lem:c2c}
  If \ref{tr:dml} - \ref{tr:rdur} hold, then after GST, any suspicion between two correct replicas is filtered in the \susmonitor.
\end{lemma}
\begin{proof}
  The \susmonitor filters suspicions as explained in \cref{sec:suspicion}:
  \textit{For each round, the \susmonitor retains only the first suspicion, discarding suspicions from later protocol phases.
  Additionally, if the leader raises suspicion for a message in round $i$, suspicions against a delayed proposal timestamp in round $i+1$ are also filtered.}\\
  Assume that the latency between correct replicas $\na$ and $\nb$ has been measured and recorded in the latency matrix after GST.
  According to our system model in \cref{sec:prelim}, after GST, the latency between \na and \nb lies in the interval $[\lat[a]{\na}{\nb}, \df{\lat[a]{\na}{\nb}}]$.
  Especially, for the recorded latency $\df{\lat[a]{\na}{\nb}} < \df{\lat[r]{\na}{\nb}}$ holds.
  Thus, any message $m$ sent from \na to \nb after GST, will have a latency, less than $\df{\lat[r]{\na}{\nb}}$.
  It follows from \ref{tr:dml} that if \na is the leader, a message $m$ sent directly after the proposal timestamp, will be received within $\df{\lat[r]{\nl}{\nb}}=\df{\dm}$.
  For other message, due to \ref{tr:dm}, the expected message duration for $d_m$ is set to $\lat[r]{\na}{\nb}$ after another message $m'$.
  Thus, if a suspicion \susp{\nb}{\na} is raised, based on a delayed message $m$, \na will also have raised a suspicion \susp{\na}{\nc} based on some message $m'$ enabling sending $m$.
  Thus, suspicion \susp{\nb}{\na} will be filtered.

  If a correct leader \nl is suspected based on Condition (a) in \cref{tab:suspicion}, the leader did not start a new round on time.
  According to \ref{tr:rdur} this means \nl did not receive a quorum of messages in time.
  Thus, \nl will have raised a suspicion in the previous round and the suspicion against \nl will be filtered.

  Finally, a correct replica \na will only raise \susp{\na}{\nb} based on Condition (c) in \cref{tab:suspicion} if a suspicion \susp{\nb}{\na} was already registered.
  If \nb is correct this suspicion \susp{\nb}{\na} will be filtered and \susp{\na}{\nb} will not be raised.
\end{proof}

%% file: tex/security_analysis_app.tex
\section{\optitree Correctness}
\label{sec:security_analysis}

In this section, we prove correctness properties for \optitree.
We note that \optitree uses a different version of the \susmonitor and computes the candidate set differently, as explained in \cref{sec:tree-suspicion}.
We therefore need to prove \ref{p:candidates} again as \ref{pt:candidates}.
Similarly, to show \ref{pt:no-susp}, we have to show that the timeout conditions \ref{tr:dml}-\ref{tr:rdur} hold for \optitree.
\optitree provides the following guarantees:
\begin{enumerate}[label=CT\arabic*]
  \item\label{pt:candidates} There are always enough candidates available to form a tree. (\cref{thm:treeform})
  \item\label{pt:suspicions} If a tree configuration fails and the system reverts to a star topology, sufficient suspicions will be recorded to update the candidate set and invalidate the failed tree configuration. (\cref{lem:treesuspicions} \&~\ref{lem:incu})
  \item\label{pt:no-susp} After GST, any suspicion between two correct replicas is filtered in the \susmonitor. (\cref{lem:tree-c2c})
  \item\label{pt:limit} After GST, at most $2\f$ reconfigurations are needed to form a correct tree. Specifically, if $\tn < \f$ replicas are faulty, at most $2\tn$ reconfigurations are needed. (\cref{thm:2t})
\end{enumerate}

\paragraph{\ref{pt:candidates} -- Sufficient Candidates}
Faulty replicas can cast suspicions on non-faulty replicas to exclude them from being selected as internal nodes.
However, \optitree can continue to form and reconfigure to new trees if there are enough candidate replicas to select the internal nodes.
The process for selecting candidate internal nodes is detailed in \cref{sec:tree-suspicion}.
It involves constructing $\MG$ from the suspicion graph \G.
\cref{thm:treeform} shows that enough candidate replicas are always available to form a tree.

\begin{theorem}
  \label{thm:treeform}
  There are always enough candidate replicas available for selecting the internal nodes to form a tree.
\end{theorem}

\begin{proof}
  A tree configuration of size $\n \ge 3\f+1$ with a branch factor $\bn \leq \sqn$ requires at most $\sqn+1$ replicas as internal nodes.
  As explained in \cref{sec:suspicion}, the graph $\G$ always contains an independent set~($\IS$) of at least $\n-\f$ replicas.
  Then, we construct $\MG$ and $\T$ from $\G$, where $\MG$ is a maximum set of disjoint edges and $\T$ is a set of nodes that form a triangle in $\G$ with an edge in $\MG$.
  For any edge in $\MG$, at most one of the two vertices is in $\IS$.
  Since $\MG$ is disjoint, each vertex belongs to at most one edge in $\MG$.
  Moreover, for any edge in $\MG$ that forms a triangle in $\G$, only one of the three vertices forming the triangle can be in $\IS$.

  Due to the maximality of $\MG$, any edge in $\MG$ can be part of at most one triangle involving a replica in $\T$.
  Therefore, if at most $\f$ replicas are outside $\IS$, the number of vertices in the independent set adjacent to edges in $\MG$ or in $\T$ is at most $\f$.
  Consequently, the remaining replicas in $\IS$ (at least $\f+1$) are available as candidates for internal nodes.

  For $\n \ge 13$, we have $\sqn < \f$, implying that for configuration sizes of 13 and above, there are always enough candidate replicas available for selecting internal nodes.
\end{proof}

\paragraph{\ref{pt:suspicions} -- Suspicions are raised}
It follows from \cref{lem:suspicions}, that suspicions are raised if a tree configuration does not collect votes in time.
The following Lemmas show additional properties on which suspicions are raised.
This ensures that the candidate set is updated and the failed tree configuration is no longer valid.
Additionally, we use the properties about what suspicions are raised in the proof of \cref{thm:2t} to show that after enough trees have failed, faulty replicas are excluded from the candidate set.

\begin{lemma}
  \label{lem:treesuspicions}
  When a tree fails, either (1) one suspicion between internal nodes is raised, or (2) $\un+1$ leaves, not in $\Crash$ get suspected.
\end{lemma}
\begin{proof}
  A tree fails, if the replica at the root fails, or cannot collect $\n-\f$ votes in time.
  If the root fails, it is suspected by other internal nodes, matching (1) from the Lemma.
  Otherwise, the root waits for aggregates from internal nodes, which together aggregate more than $\n-\f+\un$ votes.
  Thus, there are two possibilities. Either at least $\un+1$ leaves are missing and suspected by their parents, matching (2), or one internal node is either missing or did send a faulty aggregate, and is suspected by the root, matching (1).
\end{proof}

\begin{lemma}
  \label{lem:incu}
  If a tree failed, either: (i) $|\MG|$ increases, or (ii) $|\T|$ increases and $|\MG^n|$ stays unchanged.
\end{lemma}

\begin{proof}
  All suspicions are added to the graph $\G$.
  In case (1) of Lemma~\ref{lem:suspicions}, the new suspicion can be added to $\MG$ and $\T$ remains unchanged.
  In case (2), $\un+1$ edges are added to $\G$.
  These new edges include $\un+1$ different leaves from the failed tree. These leaves, can be already adjacent to an edge in $\MG$ or part of $\T$.
  There are three cases:
a) one leaf was not part of $\MG$ or $\T$. The new edge can be added to $\MG$ giving case (i) of the Lemma.
b) one leaf was part of $\T$. The new edge can be added to $\MG$ and $|\T|$ is reduced by 1, giving case (i) of the Lemma.
(3) all leaves are part of $\MG$, but since $|\MG|\leq \un$, at least two suspected leaves are already connected in $\MG$. This allows to either replace one edge in $\MG$ by two new ones, or add a new node to $\T$.
\end{proof}

%% file: tex/reconfig_proof_app.tex
\paragraph{\ref{pt:no-susp} -- Eventually no false suspicions}
\cref{lem:tree-c2c} proves \ref{p:no-susp}, ensuring that suspicions which are not filtered after GST are useful.

\begin{lemma}
  \label{lem:tree-c2c}
  After GST, any suspicion between two correct replicas is filtered in the \susmonitor.
\end{lemma}
\begin{proof}
  We need to show that the timeout requirements \ref{tr:dml} - \ref{tr:rdur} hold for \optitree.
  The Lemma then follows from \cref{lem:c2c}.

  We show that \ref{tr:dml} - \ref{tr:rdur} hold by showing how the timeouts (or durations) for the different messages are set.
  The different messages in the tree are:
  \begin{itemize}
    \item Propose messages from the root \nr to the intermediate nodes.
    \item Forwarded Propose messages from intermediate nodes to the leaf nodes.
    \item Vote messages from the leaf nodes to the intermediate nodes.
    \item Aggregated Vote messages from the intermediate nodes to the root.
  \end{itemize}

  The message timeout  $d_{m_p}$ for a Propose message from the root is set to the latency of the used link, ensuring \ref{tr:dml}.

  Forwarded Propose messages are sent directly after receiving the Propose message and Vote messages are sent directly after receiving a forwarded Propose.
  The message timeout for these messages is set by simply adding the link latency to the timeout of the previous message.
  Thus \ref{tr:dm} holds, since for a Forwarded Propose message, $m'$ is the Propose message, and for a Vote message, $m'$ is the Forwarded Propose message.

  An intermediate nodes $\nin$ sends the Aggregated Vote message after receiving Vote messages from all its children, or after waiting for the timeout for all these Vote messages.
  Thus, the timeout for the Aggregated Vote message is set according to \ref{tr:dm}, where $m'$ is the Vote message with largest timeout.

  Finally, we note that $\rdur$ is set according to \ref{tr:rdur}., where \textit{quorums} is the set $\Intermediate_k$ defined in the score function from \cref{def:score} in \cref{sec:ol-trees}.

  We note that in \optitree, we actually omit raising suspicions on Forwarded Propose messages.
  The reason for this is that, if the Vote message is delayed and the Leaf node is suspected by the intermediate node, then the leaf node will anyway reciprocate the suspicion.
  The intermediate node cannot omit the suspicion, if the Vote is delayed, without itself being suspected by the root.
  Additionally, when calculating $d_m$ for Vote messages, we set the timeout based on when the Forwarded Propose was sent, rather than the Proposal timestamp.
\end{proof}

\paragraph{\ref{pt:limit} -- Limiting the number of reconfigurations}
To prove our main theorem, about the number of reconfigurations needed to find a correct tree,
we consider new subsets $\MG^n\subset\MG$, $\T^n\subset\T$, and $\Crash^n\subset\Crash$ which are added to the respective sets after GST.
Lemma~\ref{lem:found} shows that the size of these data-structures reflects the number of faulty replicas removed from the candidate set.
Lemma~\ref{lem:incu} shows that for any two failed trees, the size of these data-structures increases by at least one.
Lemmas~\ref{lem:nodecu} and~\ref{lem:oldremove} show that other mechanisms do not reduce the size.
We denote the total size of these structures as $\tn^n=|\Crash^n|+|\MG^n|+|\T^n|$.
Remember that $\tn\leq \f$ is the number of actually faulty replicas in the system.

\begin{lemma}
  \label{lem:found}
  If $\tn^n=\tn$, then a working configuration is found.
\end{lemma}
\begin{proof}
  After GST, no two correct replicas raise suspicions on each other and all suspicions on correct replicas are reciprocated. Therefore, if a replica was added to $\Crash^n$ it is indeed faulty.
  Further, if a new edge was added to $\G$, then at least one of the adjacent replicas is faulty.
  Thus, for every edge in $\MG^n$, at least one of the adjacent nodes is faulty, and for each triangle build from a replica in $\T^n$ and an edge in $\MG^n$, at least two of the adjacent replicas are faulty.
  Finally, for a triangle in $\G$ containing a replica in $\T^n$ and an edge in $\MG$ but not in $\MG^n$, at least one of the adjacent replicas is faulty.
  Thus $\tn^n$ many faulty replicas have been removed from the candidate set for internal nodes.
\end{proof}

\begin{lemma}
  \label{lem:incugst}
  If a tree failed after GST, either: (i) $|\MG^n|$ increases, or (ii) $|\T^n|$ increases and $|\MG^n|$ stays unchanged.
  Also, either $\tn'$ increases, or $\tn'$ stays constant and $\T^n$ decreases.
\end{lemma}
This Lemma is a variation of \cref{lem:incu} and the proof follows the same schema.

\begin{proof}
  All suspicions are added to the graph $\G$.
  In case (1) of \cref{lem:suspicions}, the new suspicion can be added to $\MG^n$ and $\T$ remains unchanged, increasing $\tn'$.
  In case (2), $\un+1$ edges are added to $\G$.
  These new edges include $\un+1$ different leaves from the failed tree. These leaves, can be already adjacent to an edge in $\MG$ or part of $\T$.
  There are three cases:
a) one leaf was not part of $\MG$ or $\T$. The new edge can be added to $\MG^n$ giving case (i) of the Lemma.
b) one leaf was part of $\T$. The new edge can be added to $\MG^n$ and $|\T|$, possibly also $|\T^n|$ is reduced by 1, giving case (ii) of the Lemma. In this case, $\tn'$ stays constant.
(3) all leaves are part of $\MG$, but since $|\MG|\leq \un$, at least two suspected leaves are already connected in $\MG$. This allows to either replace one edge in $\MG$ by two in $\MG^n$, or add a new node to $\T^n$. Also here $\tn'$ increases.
\end{proof}

If a new suspicion is raised between two nodes in $G$, it is first treated as a two way suspicion and added to $\G$. If no reciprocation happens during the next $\f+1$ views, then the suspected replica is added to $\Crash$.
The next Lemma proves that these changes do not decrease $\tn'$.

\begin{lemma}
\label{lem:nodecu}
If nodes are added to $\Crash$ due to not reciprocated suspicions, this does not change $\tn'$, and does not increase $\T^n$.
\end{lemma}

\begin{proof}
Clearly, if an edge is removed from $\MG^n$, then a replica is added to $\Crash^n$.
In some cases the removal removing an edge from $\MG$ causes that a replica is removed from $\T^n$ and instead, a new edge is added to $\MG^n$.
Thus $\tn'$ does not change.
\end{proof}

All suspicions added after GST include at least one faulty replica. Thus these suspicions leave an independent set of $\n-\f$ correct replicas.
However, suspicions made before GST may be successively removed after GST.

\begin{lemma}
\label{lem:oldremove}
If suspicions made before GST are discarded, removing replicas from $\Crash$, or removing edges from $\G$, $\tn'$ does not decrease, and $\T^n$ does not increase.
\end{lemma}
\begin{proof}
Clearly this does not reduce $\Crash^n$ and $\MG^n$.
However, if an edge in $\MG$ is removed, it can be that a replica in $\T^n$ is no longer part of a triangle.
However, the replica in $\T^n$ is adjacent to an edge added after GST. This edge can now be added to $\MG^n$, leaving $\tn'$ unchanged.
\end{proof}

\begin{theorem}
\label{thm:2t}
After GST, \optitree finds a working configuration after at most $2\tn$ reconfigurations.
\end{theorem}
\begin{proof}
Based on Lemma~\ref{lem:found}, we only have to show that after $2\tn$ reconfigurations, $\tn'=\tn$.
According to Lemma~\ref{lem:incu}, after a reconfiguration, either $\tn'$ increases, or $\tn'$ stays constant and $\T^n$ decreases.
Since $\T^n$ can only decrease by $k$ after $\T^n$ and $\tn'$
have been increased by $k$. Thus at most $\tn'$ reconfigurations can leave $\tn'$ unchanged.
\end{proof}